\documentclass[12pt]{article}

\usepackage{lmodern}

\usepackage{amssymb}
\usepackage{amsmath}
\usepackage{amsthm}
\usepackage{bbm}

\usepackage{pdflscape}

\usepackage{tikz}
\usetikzlibrary{positioning}
\usetikzlibrary{arrows.meta,shapes,calc}
\usepackage{geometry}
\usepackage{subcaption}
\usepackage{float}

\usepackage{graphicx}

\usepackage{booktabs}
\usepackage{multirow}
\usepackage{setspace} 

\RequirePackage[utf8]{inputenc} 
\usepackage{xcolor}
\usepackage{csquotes}  
\usepackage{enumerate}
\usepackage{comment}

\usepackage[colorlinks=false]{hyperref}
\hypersetup{colorlinks, citecolor=blue, urlcolor=blue, linkcolor=blue} 
\usepackage{bookmark} 

\usepackage{natbib}

\theoremstyle{plain} \newtheorem*{definition*}{Definition}
\theoremstyle{plain} \newtheorem{assumption}{Assumption}

\newtheorem{assumptionprime}{Assumption}

\newtheorem{proposition}{Proposition}
\newtheorem{lemma}{Lemma}
\newtheorem{corollary}{Corollary}
\theoremstyle{definition} 

\DeclareMathOperator{\Cov}{Cov}

\newcommand{\DIIV}{\mathrm{DIIV}}
\newcommand{\RF}{\mathrm{RF}}
\newcommand{\FS}{\mathrm{FS}}

\title{Toggling the Defiers to Relax Monotonicity: \\ The Difference-in-Instrumental-Variables Estimand}%

\author{Johann Caro-Burnett\thanks{\protect\linespread{1}\protect\selectfont Hiroshima University. \href{mailto:johanncb@hiroshima-u.ac.jp}{johanncb@hiroshima-u.ac.jp.}} }

\date{\today}

\begin{document}
\maketitle

\vspace{-1cm}
\begin{abstract}\footnotesize
\noindent Standard instrumental variables (IV) methods identify a Local Average Treatment Effect under monotonicity, which rules out defiers. In many empirical environments, however, distinct instruments may induce heterogeneous and even opposing behavioral responses. This paper introduces the Difference-in-Instrumental-Variables (DIIV) estimand, which exploits two instruments with opposing compliance patterns to recover a point-identified and behaviorally interpretable causal effect without imposing monotonicity. The estimand yields a convex combination of the marginal treatment effects on compliers and defiers, with weights reflecting differential shifts in treatment take-up across instruments. When monotonicity holds, DIIV coincides with the standard IV estimand. The approach can be implemented using simple linear transformations and standard two-stage least squares procedures. Applications using replication data illustrate its applicability in practice.

\bigskip

\end{abstract}

\vfill

\pagebreak

\section*{Highlights}

\begin{itemize}
    \item This paper develops a method to estimate a Local Average Treatment Effect (LATE) when defiers are present.  
    \item The \textit{Difference-in-Instrumental-Variables} (DIIV) is introduced: a framework that leverages the contrast between two instruments to measure a convex combination of complier and defier effects. 
    \item Instead of imposing monotonicity, which rules out psychological reactance, the DIIV only requires two instruments that shift the shares of compliers and defiers in opposite directions. This feature can, arguably, be satisfied by design.
    \item A \textit{micro-foundation} for treatment take-up, aligned with information economics, is developed.
    \item It is shown that a simple variable transformation allows estimation and inference using standard tools such as two-stage least squares (2SLS).  
    \item The baseline design involves two separate encouragements in parallel experiments, but the DIIV framework extends naturally to more general settings. 
    \item Using replication data from prior studies, we illustrate how the DIIV can be applied in practice.  
\end{itemize}

\pagebreak

\section{Introduction}

A central insight of \citet{angrist1996identification} is that instrumental variables (IV) methods identify causal effects under a set of assumptions. Specifically, the IV estimand recovers a Local Average Treatment Effect (LATE) for compliers when four conditions hold: relevance, exogeneity, exclusion, and monotonicity. The monotonicity assumption requires that the instrument shifts treatment take-up in the same direction for all units, thereby ruling out the existence of defiers.

In many real-world settings, however, defiers are both plausible and empirically salient. This may be especially prominent in settings involving \emph{divisive, controversial, or authority-laden domains} such as religion, politics, or, more recently, the deployment of artificial intelligence systems. Psychological research on reactance \citep{Brehm1966,MironBrehm2006} documents that some individuals respond to encouragements or recommendations by actively resisting them. Thus, defiers are not rare anomalies but systematic behavioral responses that cannot simply be ignored. Consequently, in such environments, assuming monotonicity may be inappropriate and inference may become difficult to interpret.

Importantly, environments in which behavioral responses differ across instruments are not confined to controversial domains. They arise whenever distinct encouragements operate through different incentives, framings, informational content, or authority structures. Such designs are common in experimental and quasi-experimental work featuring multiple treatment arms, varying encouragement intensities, competing messages, or alternative policy thresholds. In these settings, the possibility that some individuals comply with one instrument but resist another is a structural feature of the design rather than a pathological exception.

This paper introduces the \emph{Difference-in-Instrumental-Variables} (DIIV) estimand, which directly incorporates the existence of defiers into causal analysis. The core idea is to use \emph{two instruments that shift compliance and defiance in opposite directions}. By contrasting how these instruments affect treatment take-up, it becomes possible to identify an interpretable causal effect even when monotonicity fails.\footnote{\label{fn:riddle} The logic parallels a classic two-guardians riddle, popularized in the film \emph{Labyrinth} and in logic puzzles. A traveler faces two doors, one safe and one unsafe, guarded by one truth-teller and one liar, and may ask only one question. By asking a logically inverted question (such as what the other guardian would say) one can recover the correct door regardless of which guardian is addressed. In a similar fashion, contrasting two instruments with opposite compliance patterns allows identification that does not depend on knowing whether a given unit is a complier or a defier. See Proposition \ref{result:2sls} for a DIIV representation that parallels the mathematical-logical solution to the two-guardians riddle.} Our setup is developed with a focus on experimental design, where researchers are aware that potential defiers exist and may deliberately induce different shifts in compliance and defiance by designing distinct instruments.

Crucially, the proposed estimand does not require researchers to observe or classify behavioral types. It relies only on variation across instruments with distinct behavioral orientation, a feature already present in many empirical designs. As a result, DIIV can be implemented in settings that already contain multiple encouragement arms or alternative policy levers, without imposing additional structure beyond standard IV conditions. More generally, the approach applies whenever two instruments induce opposing behavioral responses in compliance and defiance while satisfying relevance, exogeneity, and exclusion.

In addition, we develop a simple \emph{micro-foundation} for treatment take-up that explains why individuals from the same population may behave differently depending on the instrument. Different instruments convey different information, incentives, or framings, leading the same individual to take the treatment under one instrument and not under another. This micro-foundation clarifies why marginal treatment effects may depend on \textit{unobserved types} rather than on \textit{observed take-up behavior} alone, and provides an explicit behavioral interpretation of the objects identified by DIIV.

A key result of the paper is that the DIIV estimand always yields a \emph{convex combination} of the marginal treatment effects on compliers and defiers, providing an informative and interpretable causal object even when monotonicity fails.\footnote{See Proposition~\ref{result:maininden} for the formal expression.} In contrast to standard multi-instrument 2SLS settings, where weights may be negative or lack a clear behavioral interpretation \citep{mogstad2021causal}, the DIIV weights correspond to population shares of marginal behavioral responses. In this sense, the estimand identifies the treatment effect for units whose take-up behavior responds \emph{differentially} across the two instruments. When combined with plausible sign restrictions or stability conditions on behavioral groups, it serves as a transparent \emph{location statistic} for the underlying marginal effects.

Because the object is defined in terms of differential responsiveness across instruments, it has a natural interpretation in design-based work: it measures the average outcome change associated with shifting from a weaker-compliance instrument to a stronger-compliance one. This interpretation is particularly relevant in environments where researchers compare alternative encouragement strategies, messaging regimes, or policy intensities. Under plausible sign restrictions or stability conditions on behavioral groups, DIIV also serves as a \emph{location statistic} for the underlying marginal effects. DIIV therefore extends the IV toolkit to settings where behavioral heterogeneity induced by incentives, framing, or reactions to authority is central, while preserving a point-identified and behaviorally interpretable effect.


\subsection*{Related Literature}

The instrumental variables framework classifies individuals into four types: always-takers, never-takers, compliers, and defiers. The standard LATE framework assumes monotonicity, thereby ruling out defiers and enabling point identification of complier effects \citep{imbens1994,angrist1996identification}. While analytically convenient, monotonicity can be substantively strong. A large literature has therefore examined its empirical content and implications.

One strand develops empirical diagnostics for IV validity and related identifying restrictions, including inequality restrictions, moment conditions, and overidentification tests when multiple instruments are available \citep{hansen1982,huber2015,kitagawa2015,mourifie2017}. These approaches clarify testable implications of the IV model but do not in themselves recover interpretable effects when defiers are present.

When monotonicity fails, point identification of the LATE generally breaks down. Partial identification approaches derive bounds on average treatment effects that explicitly allow for defiers \citep{BalkePearl1997,manski1990}, with subsequent refinements incorporating covariates and additional structure \citep{richardson2010,swanson2015}. More recent contributions develop sensitivity analyses that parameterize monotonicity violations and assess robustness to different shares of defiers \citep{noack2021}, or attempt to estimate or bound the prevalence of defiers directly \citep{christy2024countingdefiersdesignbasedmodel}. These methods provide valuable information about identification limits but typically replace point identification with bounds or sensitivity regions.

A related literature proposes weaker versions of monotonicity. Stochastic monotonicity permits certain forms of non-deterministic compliance behavior while preserving interpretable weighted effects \citep{small2017}, and other structured relaxations accommodate defiers under additional assumptions \citep{dechaisemartin2017,dahl2023}. The most closely related work is \citet{mogstad2021causal}, who show that with multiple instruments and ``partial monotonicity'' (monotonicity by each instrument, not necessarily by the vector of instruments), 2SLS recovers a positively weighted average of instrument-specific LATEs. Their work highlights the interpretive complexity that can arise when aggregating instruments, even under weakened monotonicity.

In contrast to bounds, sensitivity analyses, or partial monotonicity approaches, DIIV exploits two instruments with opposing compliance patterns to recover a point-identified estimand with a direct behavioral interpretation, without imposing monotonicity. By focusing on differential responsiveness across instruments, it provides a simple and general extension of the IV framework to environments where defiance is a structural feature of the design.


\section{Illustrative Design: Two Parallel Instruments}\label{sec:parallel}

We first develop a basic setup with two instruments that encourage treatment take-up, and shift compliance and defiance in opposite directions. In later sections, we extend the analysis to more general settings. 

Let $I$ be a population of units $i\in I$. By design, the population is partitioned into two ex ante identical subpopulations, $I_1$ and $I_2$, where each group is potentially exposed to an exclusive instrument, also referred to as a nudge.\footnote{Here, ``identical'' refers to underlying type composition and potential outcomes; realized compliance behavior may differ across subpopulations due solely to different framing of each instrument.} For $j\in\{1,2\}$, the exclusive instrument in subpopulation $I_j$ is denoted $z_{i,j}\in\{0,1\}$ and is exogenously assigned with the same probability across subpopulations. Instrument $z_{i,j}$ encodes not only assignment but also the framing associated with instrument $j$.

Define $D_{i,j}(z)\in\{0,1\}$, for $z\in\{0,1\}$, as the potential treatment take-up in $I_j$, which depends both on the assignment $z$ and on the framing $j$. The observed treatment status is $d_i=D_{i,j}(z_{i,j})$, which, together with the potential outcomes $Y_i(d)$ for $d\in\{0,1\}$, yields the observed outcome
\begin{equation}
y_i = d_i Y_i(1) + (1-d_i)Y_i(0).
\end{equation}
In particular, the distribution of $y_i$ may depend on the framing $j$ only through the treatment decision $d_i$ and not through frame-specific potential outcomes.\footnote{From the notation, the exclusion restriction is satisfied; see Assumption \ref{assumption:exclu}.}

Based on how units behave in response to the instrument, they can be one of four types $\theta_i\in\{A,N,C,F\}$, which partition $I$, $I_1$, and $I_2$. Let $\pi_\theta$ denote the population mass of type $\theta$. Types $A$ and $N$ are the traditional \emph{always-takers} and \emph{never-takers}.\footnote{Types $A$ and $N$ are not essential and could be assumed to have zero mass.} Type $C$ consists of units that find value in the nudge and are \emph{potential compliers}, whereas type $F$ consists of units that find disvalue in the nudge and are \emph{potential defiers}. Depending on responsiveness to framing $j$, subpopulations $I_1$ and $I_2$ may exhibit different shares of units behaving as compliers ($d_i=z_{i,j}$) or as defiers ($d_i=1-z_{i,j}$).

More precisely, given an instrument with framing $j$, potential compliers (type $C$) take the treatment iff
\begin{equation}
v_C(j) z_{i,j} + \eta_i \ge \underline v,
\end{equation}
where $\eta_i$ is an idiosyncratic shock with an atomless distribution; $v_C(j)\ge 0$ captures responsiveness to framing $j$; and $\underline v$ is a threshold.\footnote{In a related approach, \citet{kowalski2023behaviour} exploit an ordering of the slopes of marginal treated and untreated outcome as functions of covariates to bound treatment effects for always takers.} Conditional on type $C$, units with $\eta_i\in[\underline v - v_C(j),\underline v)$ are contingent compliers, denoted $C_j^C$, with conditional mass $\phi_C^C(j)$. Those with $\eta_i\ge\underline v$ behave as contingent always-takers ($C_j^A$), with conditional mass $\phi_C^A(j)$, while those with $\eta_i<\underline v-v_C(j)$ behave as contingent never-takers ($C_j^N$), with conditional mass $\phi_C^N(j)$. These satisfy $\phi_C^C(j)+\phi_C^A(j)+\phi_C^N(j)=1$ (see Figure \ref{fig:microc}).

Similarly, given an instrument with framing $j$, potential defiers (type $F$) take the treatment iff
\begin{equation}
v_F(j) z_{i,j} + \eta_i \ge \underline v,
\end{equation}
where $v_F(j)\le 0$. Conditional on type $F$, units with $\eta_i\in[\underline v,\underline v - v_F(j))$ are contingent defiers, denoted $F_j^F$, with conditional mass $\phi_F^F(j)$. Those with $\eta_i\ge\underline v - v_F(j)$ behave as contingent always-takers ($F_j^A$), with conditional mass $\phi_F^A(j)$, while those with $\eta_i<\underline v$ behave as contingent never-takers ($F_j^N$), with conditional mass $\phi_F^N(j)$. These satisfy $\phi_F^F(j)+\phi_F^A(j)+\phi_F^N(j)=1$ (see Figure \ref{fig:microf}).

\begin{figure}[ht]
\centering
\begin{subfigure}{\linewidth} 
\centering
\begin{tikzpicture}[scale=1]

  \fill[gray!15] (-3,0) rectangle (3,1);

  \fill[gray!30] (3,0) rectangle (9,1);

  \node at (-2,0.5) {non-takers};
  \node at (8,0.5) {takers};

  \draw[<->] (-3.5,0) -- (9.5,0) node[right] {$\eta_i$};

  \draw[dashed] (3,0) -- (3,1);
  \draw (3,0.08) -- (3,-0.08);
  \node[below] at (3,-0.1) {$\underline{v}$};
\end{tikzpicture}
\subcaption{$z=0$.} \label{fig:microz0}
\end{subfigure}

\vspace{0.4cm}

\begin{subfigure}{\linewidth} 
\centering
\begin{tikzpicture}[scale=1]

  \fill[gray!15] (-3,0) rectangle (2,1);

  \fill[gray!60] (1,0) rectangle (3,1);

  \fill[gray!30] (3,0) rectangle (9,1);

  \node at (-2,0.5) {non-takers};
  \node at (2,0.5) {compliers};
  \node at (8,0.5) {takers};

  \draw[<->] (-3.5,0) -- (9.5,0) node[right] {$\eta_i$};

  \draw[dashed] (1,0) -- (1,1);
  \draw[dashed] (3,0) -- (3,1);

  \draw (1,0.08) -- (1,-0.08);
  \draw (3,0.08) -- (3,-0.08);
  \node[below] at (1,-0.1) {$\underline{v} - v_C(g)$};
  \node[below,yshift=-0.75ex] at (3,-0.1) {$\underline{v}$};
\end{tikzpicture}
\subcaption{$z=1$, $\theta = C$ potential compliers.} \label{fig:microc}
\end{subfigure}

\vspace{0.4cm}

\begin{subfigure}{\linewidth} 
\centering
\begin{tikzpicture}[scale=1]

  \fill[gray!15] (-3,0) rectangle (3,1);

  \fill[gray!60] (3,0) rectangle (4.5,1);

  \fill[gray!30] (4.5,0) rectangle (9,1);

  \node at (-2,0.5) {non-takers};
  \node at (3.75,0.5) {defiers};
  \node at (8,0.5) {takers};

  \draw[<->] (-3.5,0) -- (9.5,0) node[right] {$\eta_i$};

  \draw[dashed] (3,0) -- (3,1);
  \draw[dashed] (4.5,0) -- (4.5,1);

  \draw (3,0.08) -- (3,-0.08);
  \draw (4.5,0.08) -- (4.5,-0.08);
  \node[below,yshift=-0.75ex] at (3,-0.1) {$\underline{v}$};
  \node[below] at (4.7,-0.1) {$\underline{v} - v_F(g)$};
\end{tikzpicture}
\subcaption{$z=1$, $\theta = F$ potential defiers.} \label{fig:microf}
\end{subfigure}
\caption{Thresholds for treatment take-up for potential compliers and potential defiers.}
\end{figure}
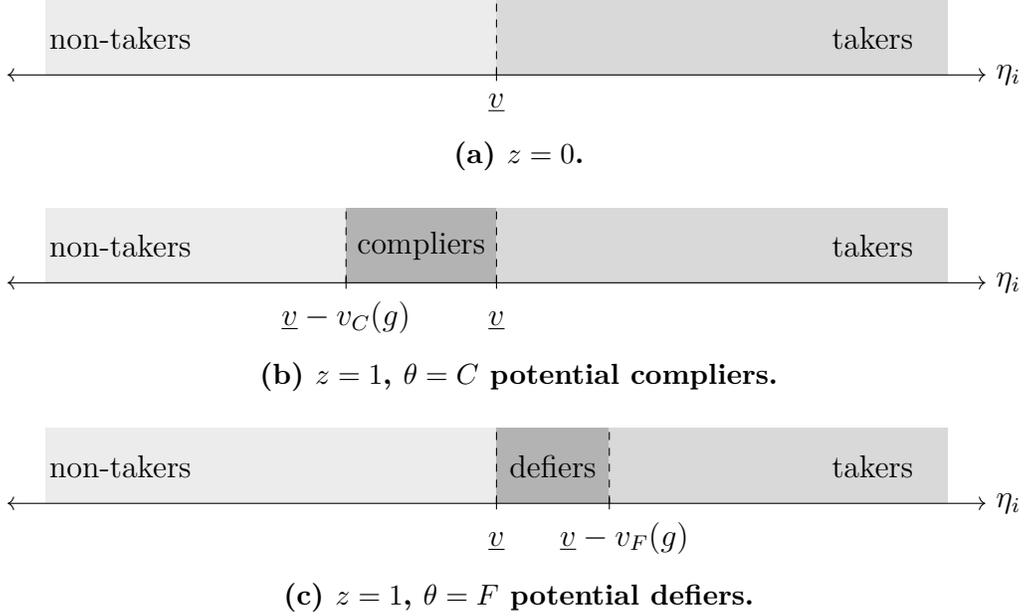

Thus, for each framing $j\in\{1,2\}$, we obtain a partition
\[
I_j=\{A_j,N_j,C_j^C,C_j^A,C_j^N,F_j^F,F_j^A,F_j^N\}.
\]
We pay special attention to the unconditional shares of behavioral compliers and behavioral defiers, which depend on the framing. We denote these by
\begin{equation}
p_C^j := \pi_C \phi_C^C(j) \text{ and } p_F^j := \pi_F \phi_F^F(j).
\end{equation}

Although we do not impose the existence of defiers (i.e., $F=\emptyset$ is allowed), we develop a method to estimate local average treatment effects when they do exist. To estimate the effect of treatment take-up on outcomes, we impose the following assumptions.

\begin{assumption}\label{assumption:exo}
\textbf{Exogeneity:} For each instrument $z_j$, $j\in\{1,2\}$,
\begin{equation}
z_j \;\perp\; \left( Y_i(0),Y_i(1),D_{i,j}(0),D_{i,j}(1) \right).
\end{equation}
\end{assumption}

\begin{assumption}\label{assumption:exclu}
(i) \textbf{Exclusion:} Conditional on treatment take-up and within each subpopulation $j$, the potential outcome does not depend on assignment:
\begin{equation}
Y_i(d\mid z_{i,j}=1) = Y_i(d\mid z_{i,j}=0).
\end{equation}
(ii) \textbf{Type-invariant expected marginal effects:} For any $\theta\in\{A,N,C,F\}$,
\begin{equation}
\tau_\theta := \mathbb{E}\big[Y_i(1)-Y_i(0)\mid \theta_i=\theta\big]
\end{equation}
does not depend on $j$.
\end{assumption}

The next assumption replaces monotonicity while also serving as the relevance condition.

\begin{assumption}\label{assumption:main}
(i) \textbf{Opposing shifts:} $p_C^1 \ge p_C^2$ and $p_F^1 \le p_F^2$.
(ii) \textbf{Relevance:} At least one inequality holds strictly:
\begin{equation}
(p_C^1 - p_C^2) + (p_F^2 - p_F^1) > 0.
\end{equation}
\end{assumption}

For $j\in\{1,2\}$, define the reduced form and first stage as
\begin{equation}
RF_j := \mathbb{E}[y_i \mid z_{i,j}=1] - \mathbb{E}[y_i \mid z_{i,j}=0]
\end{equation}
and
\begin{equation}
FS_j := \mathbb{E}[d_i \mid z_{i,j}=1] - \mathbb{E}[d_i \mid z_{i,j}=0].
\end{equation}

\begin{lemma} \label{result:rffs}
Under Assumption \ref{assumption:exclu},
\begin{equation}
RF_j = p_C^j \tau_C - p_F^j \tau_F \text{ and }
FS_j = p_C^j - p_F^j.
\end{equation}
\end{lemma}

\begin{proof}
See the Appendix.
\end{proof}

Lemma \ref{result:rffs} shows how the standard IV estimand $RF_j/FS_j$ is biased when $p_F^j>0$. We therefore propose an alternative estimand that exploits differential shifts in compliance and defiance.

\begin{definition*}
Let us define the DIIV estimand as:
\begin{equation}
\tau_{\mathrm{DIIV}} := \frac{RF_1 - RF_2}{FS_1 - FS_2},
\end{equation}
where $FS_1 \neq FS_2$, which is ensured by Assumption \ref{assumption:main}.
\end{definition*}

\begin{proposition}\label{result:maininden}
If Assumptions \ref{assumption:exo}, \ref{assumption:exclu} and \ref{assumption:main} hold, then the DIIV estimand identifies:
\begin{equation}
\tau_{\mathrm{DIIV}} = \lambda \tau_C + (1-\lambda)\tau_F,
\end{equation}
where
\begin{equation}
\lambda
=
\frac{p_C^1 - p_C^2}{(p_C^1 - p_C^2)-(p_F^1 - p_F^2)}
\in [0,1].
\end{equation}
\end{proposition}

\begin{proof}
See the Appendix.
\end{proof}

\begin{corollary}\label{result:specialcases}
If Assumptions \ref{assumption:exo}, \ref{assumption:exclu} and \ref{assumption:main} hold, then:
\begin{equation}
\text{(i) If } p_F^1 = p_F^2, \ \tau_{\mathrm{DIIV}} = \tau_C
\text{ and (ii) If } p_C^1 = p_C^2, \ \tau_{\mathrm{DIIV}} = \tau_F.
\end{equation}
\end{corollary}

\begin{proof}
See the Appendix.
\end{proof}

Note that, under standard IV terminology and assumptions (i.e., monotonicity and the interpretation that instruments act as encouragements), a special case of part (i) is obtained when $p_F^1=p_F^2=0$.

The DIIV estimand can be interpreted as a LATE that mixes the effects on potential complier (type $C$) and potential defier (type $F$) units whose behavior shifts when moving from one framing to the other. The weights in Proposition \ref{result:maininden} are determined by how the change in framing alters the shares of contingent compliers and defiers. Thus, $\tau_{\mathrm{DIIV}}$ captures a LATE that is directly tied to the policy margin of modifying how a nudge is framed.

Note that both instruments $z_j$ could be pooled into a single instrument across the entire population $I=I_1\cup I_2$ by extending the assignment $z_{i,j}=0$ to units in the other subpopulation. In that case, defining $z_{\mathrm{pool}}=z_1+z_2$, the standard IV estimand would be
\begin{equation}
\tau_{\mathrm{pool}}
=
\frac{RF_1 + RF_2}{FS_1 + FS_2}
=
\frac{(p_C^1+p_C^2)\tau_C - (p_F^1+p_F^2)\tau_F}
     {(p_C^1+p_C^2)-(p_F^1+p_F^2)},
\end{equation}
which is a weighted but non-convex combination of $\tau_C$ and $\tau_F$. Therefore, pooling alone does not resolve the bias induced by defiers.

Instead, consider a \emph{mistakenly} coded instrument $z_2^+ = 1 - z_2$ (i.e., a toggle) and pool it with $z_1$. By flipping the directive of the second instrument, compliers become defiers and defiers become compliers relative to that framing. This change in roles enables a differential comparison across framings that restores identification.\footnote{This role reversal is analogous to the solution of the previously mentioned two-guardian riddle in Footnote \ref{fn:riddle}. Asking one guardian about the answer the other would give effectively flips the directive: truth-telling is converted into lying and vice versa. In the present setting, recoding the instrument as $z_2^+=1-z_2$ similarly toggles the behavioral mapping of the instrument, transforming compliers into defiers and defiers into compliers relative to the original framing.}

\begin{proposition}\label{result:fliploss}
Consider $z_2^+ = 1 - z_2$. Then
\begin{equation}
RF_2^+ = - RF_2 \text{ and } FS_2^+ = - FS_2.
\end{equation}
Define the pooled-and-flipped instrument $z_{\mathrm{pool\&flip}} = z_1 + z_2^+$. The standard IV estimand using $z_{\mathrm{pool\&flip}}$ satisfies
\begin{equation}
\tau_{\mathrm{pool\&flip}}
=
\frac{RF_1 - RF_2}{FS_1 - FS_2}
=
\tau_{\mathrm{DIIV}}.
\end{equation}
\end{proposition}

\begin{proof}
See the Appendix.
\end{proof}


\section{Estimation and Inference}

The DIIV estimand is an unbiased convex combination of $\tau_C$ and $\tau_F$. From its definition, it can be estimated using differences in sample means:
\begin{equation}
\widehat{\tau}_{\mathrm{DIIV}}
=
\frac{(\bar y_1^1-\bar y_1^0)-(\bar y_2^1-\bar y_2^0)}
     {(\bar d_1^1-\bar d_1^0)-(\bar d_2^1-\bar d_2^0)},
\end{equation}
where $\bar y_j^z$ denotes the sample mean of $y_i$ in subpopulation $j$ with $z_{i,j}=z$, and similarly for $\bar d_j^z$. While this estimator is straightforward to compute, inference based directly on this ratio is less convenient.

The next result shows that the DIIV estimand admits a numerically equivalent two-stage least squares (2SLS) representation, which allows for conventional inference using standard tools.

\begin{proposition}\label{result:2sls}
Let $h_i\in\{0,1\}$ indicate the framing, with $h_i=1$ if unit $i\in I_1$ and $h_i=0$ if unit $i\in I_2$. Let $z_{\mathrm{pool}}=z_1+z_2\in\{0,1\}$ denote the within-frame assignment. Construct the composite instrument
\begin{equation}
w_i = z_{\mathrm{pool},i} h_i + (1-z_{\mathrm{pool},i})(1-h_i).
\end{equation}
Estimate the 2SLS system
\begin{align}
\text{First stage:}\quad
& d_i = \alpha_d + \beta w_i + \gamma_d h_i + e_{d,i}, \\
\text{Second stage:}\quad
& y_i = \alpha_y + \tau \hat d_i + \gamma_y h_i + e_{y,i}.
\end{align}
Then, the 2SLS estimand satisfies
\begin{equation}
\tau_{\mathrm{2SLS}}
=
\frac{RF_1 - RF_2}{FS_1 - FS_2}
=
\tau_{\mathrm{DIIV}}.
\end{equation}
\end{proposition}

\begin{proof}
See the Appendix.
\end{proof}

This 2SLS representation is not unique. However, it is convenient and elegant. The instrument $w_i$ can be written as $w_i = 1 - (z_{\mathrm{pool},i} \oplus h_i)$, where $\oplus$ denotes the \textit{exclusive-or operator}.\footnote{The role of the exclusive-or (XOR) operator mirrors the logical structure underlying the two-guardian riddle from Footnote \ref{fn:riddle}. The solution is achieved by effectively applying an XOR operation: truthfulness and falsity cancel out because the answer depends on whether exactly one of the two conditions holds. Asking one guardian what the other would say toggles the truth value twice, producing a determinate answer. Analogously, the XOR operator here formalizes the idea that alignment versus misalignment between the pooled instrument and the framing matters only through their exclusive disagreement, allowing the directive to be flipped in a way that restores identification.} It equals one when the directive of the pooled instrument aligns with the framing and zero otherwise. This representation provides a convenient implementation of DIIV via standard 2SLS. Although elegant, this particular 2SLS construction is not robust to asymmetries or correlations in instrument assignment. In the next section, we show that while the DIIV estimand remains valid in more general settings, the 2SLS representation requires appropriate adjustments.

\section{General Framework}\label{sec:general}

We now extend the analysis to a broader setting in which two binary instruments $(z_1,z_2)\in\{0,1\}^2$ may operate jointly, each combination occurring with probability $q_{z_1z_2}$, and their intended directive for treatment take-up is explicit. To interpret how each instrument influences behavior, we associate to each $z_j$ a \emph{frame}
\begin{equation*}
f_j = (m_j, s_j),    
\end{equation*}

where $m_j$ captures the salient \emph{attribute} emphasized by the frame and $s_j\in\{-1,+1\}$ records the directive orientation: $s_j=+1$ for \textit{encouragement} and $s_j=-1$ for \textit{discouragement}.\footnote{This decomposition aligns with the distinction, common in information design and framing models, between the \emph{content} of information and its
\emph{directional} implications for choice. In Bayesian persuasion models, a sender selects an information structure that determines which attributes or dimensions of the decision problem become salient, while posterior beliefs (and hence actions) respond endogenously to this structure \citep{KamenicaGentzkow2011}. Similarly, framing models allow behavior to vary
depending on which attribute of an otherwise identical choice environment is emphasized, even when objective payoffs are unchanged \citep{TverskyKahneman1981,DellaVigna2009}. In our setting, $m_j$ captures the attribute or informational dimension activated by the instrument, while $s_j$ captures the directive orientation along that dimension (encouragement versus
discouragement). This separation allows instruments to differ in \emph{what} aspect of the choice they target, independently of \emph{how} they push behavior, a distinction that underlies the construction of DIIV based on the exclusive-or ($\oplus$) operator in Proposition \ref{result:2sls}.} We also note that since the directive's sign matters, we no longer refer to units in populations $C$ and $F$ as potential compliers and potential defiers; instead, we will refer to them as \textit{persuation-prone} and \textit{reactance-prone}, respectively.

The attribute $m_j$ can be very general, ranging from monetary incentives to purely persuasive content. In experimental settings, $s_j$ reflects the researcher’s intended direction; in observational settings, it reflects the theoretically inferred direction. We assume that $s_j$ is unambiguous and well defined. Under monotonicity, orientation plays no identifying role.\footnote{For example, by replacing an instrument $z$ by $1-z$, the sign of the ratio $\FS/\RF$ would remain unchanged.} However, when monotonicity fails, the directive becomes essential for interpreting the behavior of those who take the treatment as intended by the directive $s$ (i.e., persuation-prone) and those who do the opposite (i.e., reactance-prone).\footnote{This is particularly important when marginal effects differ between persuasion-prone and reactance-prone units. If marginal effects were homogeneous across types, the directive would not matter for interpretation.}

Each attribute $m_j$ determines how strongly units respond to the presence of the instrument, while the directive $s_j$ determines the sign of that response. For each behavioral type $\theta\in\{C,F\}$, let $\kappa_\theta(m_j)$ measure the intensity with which type $\theta$ \emph{responds to} attribute $m_j$. Persuasion-prone types ($C$) respond positively to the attribute and therefore satisfy $\kappa_C(m_j)\ge 0$, whereas reactance-prone types ($F$) respond negatively and satisfy $\kappa_F(m_j)\le 0$.

Treatment take-up follows a threshold rule in which the directed frames from both instruments jointly influence behavior:
\begin{equation}\label{eq:threshold2}
\kappa_\theta(m_1)\,s_1 z_1
\;+\;
\kappa_\theta(m_2)\,s_2 z_2
\;+\;
\eta_i
\;\ge\;
\underline v,
\end{equation}
where $\eta_i$ is an idiosyncratic shock. Thus, the frame component $m_j$ governs the strength of the effect, while $s_j$ governs its sign. Units may behave differently depending on their type and the realization of~$\eta_i$.

For each type $\theta\in\{C,F\}$ and instrument $j$, let $\varphi_\theta(j)$ denote the probability that the threshold is met when comparing $(z_j{=}1,z_{-j}{=}0)$ to $(z_j{=}0,z_{-j}{=}0)$. From the threshold rule in Equation \eqref{eq:threshold2}, the difference between these two assignments is governed solely by the increment $\kappa_\theta(m_j)s_j$ associated with setting $z_j$ from $0$ to $1$ while holding $z_{-j}=0$ fixed. Thus, type $\theta$ units respond to
instrument $j$ precisely when the idiosyncratic shock $\eta_i$ lies in the interval
\[
\eta_i \in 
\begin{cases}
[\;\underline v - \kappa_\theta(m_j)s_j,\;\underline v\;), 
& \text{if }\kappa_\theta(m_j)s_j > 0 ,\\[0.4em]
(\;\underline v,\;\underline v - \kappa_\theta(m_j)s_j\;], 
& \text{if }\kappa_\theta(m_j)s_j < 0 .
\end{cases}
\]
Then, $\varphi_\theta(j)$ is the probability mass of $\eta_i$ contained in this interval. The behavioral share of type $\theta$ shifted by instrument $j$ is then
\[
p_\theta^j = \pi_\theta\, \varphi_\theta(j),
\]
where $\pi_\theta$ is the population share of type $\theta$. The probability $p_\theta^j$ extends the notions of contingent compliance and defiance shares from the parallel-frame setting and characterizes how each instrument shifts persuasion-prone and reactance-prone groups under its specific attribute and directive. These quantities serve as the building blocks for defining first-stage contrasts and for establishing identification of the general DIIV estimand.

To estimate the causal effect of treatment take-up $d_i$ on outcome $y_i$ in the
presence of opposing behavioral responses, we impose a set of assumptions that
parallel the standard IV framework while replacing monotonicity with a
directional structure on compliance and defiance. These assumptions are denoted
with primes ($\prime$) to distinguish them from the baseline IV assumptions.

\setcounter{assumptionprime}{0}

\begin{assumptionprime}\label{assumption:exo_prime}
\textbf{Exogeneity:} The joint instrument assignment is independent of all potential outcomes and
potential treatment take-up states:
\[
(z_1,z_2)
\;\perp\;
\bigl(
Y_i(0),Y_i(1),
D_{i,1}(0),D_{i,1}(1),
D_{i,2}(0),D_{i,2}(1)
\bigr).
\]
\end{assumptionprime}

\begin{assumptionprime}\label{assumption:exclu_prime}
\textbf{(i) Exclusion.} Conditional on treatment take-up and within each
single-edge comparison, the potential outcome does not depend on the instrument
assignment:
\[
Y_i(d \mid z_{i,j}=1, z_{i,-j}=0)
=
Y_i(d \mid z_{i,j}=0, z_{i,-j}=0),
\]

for $d\in\{0,1\}$ and $j\in\{1,2\}$.

\medskip
\textbf{(ii) Type-invariant expected marginal effects.} Let
\begin{equation}
\tau_C^{\mathrm{abs}}
:= 
\mathbb{E}\!\left[Y_i(1)-Y_i(0)\mid C\right] \text{ and }
\tau_F^{\mathrm{abs}}
:= 
\mathbb{E}\!\left[Y_i(1)-Y_i(0)\mid F\right].
\end{equation}
do not depend on $j$.
\end{assumptionprime}

\begin{assumptionprime}\label{assumption:main_prime}
\textbf{(i) Opposing shifts.} The two instruments differentially shift behavioral
types in opposite directions:
\[
p_C^1 \ge p_C^2 \text{ and }
p_F^1 \le p_F^2.
\]

\medskip
\textbf{(ii) Relevance.} At least one of these inequalities holds strictly:
\[
p_C^1 - p_C^2 + p_F^2 - p_F^1 > 0.
\]
\end{assumptionprime}

Note that Assumption \ref{assumption:main_prime} is equivalent to:

\begin{equation}
\kappa_C(m_1)s_1
\;\ge\;
\kappa_C(m_2)s_2
\;\ge\;
0
\;\ge\;
\kappa_F(m_2)s_2
\;\ge\;
\kappa_F(m_1)s_1.
\end{equation}

Let $\RF_{j}^{(z)}$ and $\FS_{j}^{(z)}$ be the reduced form and first stage with respect to instrument $j$ by keeping $z_{-j}=z$ constant:

\begin{equation}
RF_{j}^{(z)} := \mathbb{E}[y_{i} \mid z_{i,j}=1,z_{i,-j}=z]-\mathbb{E}[y_{i} \mid z_{i,j}=0,z_{i,-j}=z],
\end{equation}
and
\begin{equation}
FS_{j}^{(z)} := \mathbb{E}[d_{i} \mid z_{i,j}=1,z_{i,-j}=z]-\mathbb{E}[d_{i} \mid z_{i,j}=0,z_{i,-j}=z].
\end{equation}

Next, we extend DIIV to the oriented edge differences case.

\begin{definition*}
Define the DIIV estimand
\begin{equation}
\tau_{\DIIV}
~:=~
\frac{s_1\,\RF_{1}^{(0)} - s_2\,\RF_{2}^{(0)}}
     {s_1\,\FS_{1}^{(0)} - s_2\,\FS_{2}^{(0)}}.
\end{equation}
\end{definition*}

\begin{proposition}\label{result:general-diiv}
If Assumptions \ref{assumption:exo_prime}, \ref{assumption:exclu_prime}, and
\ref{assumption:main_prime} hold, then the DIIV estimand identifies
\begin{equation}
\tau_{DIIV} = \lambda \tau_C^{\mathrm{abs}} + (1-\lambda)\tau_F^{\mathrm{abs}},
\end{equation}
where
\begin{equation}
\lambda
=
\frac{p_C^1 - p_C^2}
     {(p_C^1 - p_C^2)-(p_F^1 - p_F^2)}
\in [0,1].
\end{equation}
\end{proposition}

\begin{proof}
See the Appendix.
\end{proof}

Note that the influence of the vector $(z_1,z_2)=(1,1)$ on $d$ and $y$ does not enter the DIIV estimand; DIIV is an \emph{edge} contrast with respect to $(0,0)$. That is, it compares $(1,0)$ with $(0,0)$ and $(0,1)$ with $(0,0)$. The following Lemma is useful for understanding how these
edge contrasts are represented under directive orientations.

\begin{lemma}\label{result:aligned-edge-identity}
Let $s_1,s_2\in\{-1,+1\}$ be known constants and define the aligned instruments
\[
\tilde z_j := \frac{1-s_j}{2}+s_j z_j,\qquad j\in\{1,2\}.
\]
For $a,b\in\{0,1\}$ define the aligned-cell means
\[
y_{ab}:=\mathbb{E}[y_i\mid \tilde z_{1i}=a,\tilde z_{2i}=b],
\]
and 
\[
d_{ab}:=\mathbb{E}[d_i\mid \tilde z_{1i}=a,\tilde z_{2i}=b].
\]
Then,
\begin{equation}\label{eq:aligned-edge-diiv}
\frac{(y_{10}-y_{00})-(y_{01}-y_{00})}{(d_{10}-d_{00})-(d_{01}-d_{00})}
=
\frac{s_1\,\RF_{1}^{(0)}-s_2\,\RF_{2}^{(0)}}
     {s_1\,\FS_{1}^{(0)}-s_2\,\FS_{2}^{(0)}}.
\end{equation}
\end{lemma}

\begin{proof}
See the Appendix.
\end{proof}

Finally, we generalize Proposition \ref{result:2sls} to include asymmetric and joint
distributions of $(z_1,z_2)$ with explicit directives $s_j$ that allow for
encouragements and discouragements.

\begin{proposition}\label{result:joint2sls}
Let $(z_1,z_2)\in\{0,1\}^2$ be instruments with directive orientations $s_1,s_2\in\{-1,+1\}$, where $s_j=+1$ denotes an encouragement directive and $s_j=-1$ a discouragement directive. Let the aligned instruments be $\tilde z_j := \frac{1-s_j}{2}+s_j z_j$,

and define the variables
\begin{eqnarray*}
    x^\Delta &:=& \tilde z_1 - \tilde z_2,\\
    x^\Sigma &:=& \tilde z_1 + \tilde z_2,\\
    x^\times &:=& \tilde z_1 \tilde z_2.
\end{eqnarray*}
Then, the 2SLS of $y_i$ on $d_i$, using instrument $x^\Delta$ and covariates
$(x^\Sigma,x^\times)$ following
\begin{align}
\text{First stage:}\quad
& d_i = \alpha_d + \beta x^\Delta + \gamma_d^\Sigma x^\Sigma
      + \gamma_d^\times x^\times + e_{d,i},\\
\text{Second stage:}\quad
& y_i = \alpha_y + \tau \hat d_i + \gamma_y^\Sigma x^\Sigma
      + \gamma_y^\times x^\times + e_{y,i},
\end{align}
estimates
\[
\tau_{\text{2SLS}} = \tau_{\DIIV}.
\]
\end{proposition}

\begin{proof}
See the Appendix.
\end{proof}

\section{Applications}

We use data from three replication packages from studies regarding education \citep{barrera2019medium}, loan take-up \citep{bertrand2010s}, and vote turnout \citep{gerber2008secrecy,gerber2012ballotsecrecy}.\footnote{The data from the former two studies is publicly available. Although the data to replicate from \cite{gerber2012ballotsecrecy} is publicly available, DIIV requires a disaggregation of the intervention arms. This additional data was generously provided by Gregory Huber.} All studies are structured as large-scale randomized controlled trials (RCTs) with multiple treatment arms, which we regard as instruments. From each study, we take one immediate outcome and use it as treatment take-up ($d$) and estimate how it affected a later outcome ($y$). That is, we assume exclusion from the instruments onto $y$. For each study, we discuss why this is likely the case.


\subsection{Example 1: Cash Transfers and Long-Term Effects on Education}
\label{sec:app1_uncorr}

This section uses data from \cite{barrera2019medium} to apply DIIV to a large conditional cash transfer RCT with three encouragement arms. The study population consists of $N=17{,}309$ students from Colombia. The intervention studies the effect of cash transfers on education outcomes. The three intervention arms are: (G1) offers USD~30 every two months conditional on secondary school attendance; (G2) has the same flow as G1 but with USD~10 retained and paid at the end of the year (revenue-equivalent conditional on compliance); and (G3) has the same flow as G2 plus an additional USD~300 lump-sum if the student enrolls in tertiary education (e.g., college). Table~\ref{tab:repli1} reproduces the immediate and long-term effects on education outcomes following \cite{barrera2019medium}, Tables~3 and~4.

\begin{table}[!ht]
\caption{Replication of \cite{barrera2019medium}}
\begin{center}
    \begin{tabular}{lrccccccccccc}
          &       & \multicolumn{3}{c}{comprehensive exam} &       & \multicolumn{3}{c}{tertiary enrollment} &       & \multicolumn{3}{c}{tertiary graduation} \\
\cmidrule{3-5}\cmidrule{7-9}\cmidrule{11-13}          &       &       &       &       &       &       &       &       &       &       &       &  \\
    G1    &       & 0.02  &       & -0.01 &       & 0.00  &       & -0.00 &       & 0.00  &       & -0.01 \\
          &       & (0.01) &       & (0.01) &       & (0.01) &       & (0.01) &       & (0.01) &       & (0.00) \\
    G2    &       & 0.01  &       & -0.03** &       & 0.01  &       & -0.03** &       & 0.01  &       & 0.00 \\
          &       & (0.01) &       & (0.01) &       & (0.01) &       & (0.01) &       & (0.01) &       & (0.01) \\
    G3    &       &       & 0.01  & 0.12*** &       &       & 0.03* & 0.07*** &       &       & 0.01  & 0.06*** \\
          &       &       & (0.01) & (0.01) &       &       & (0.02) & (0.01) &       &       & (0.01) & (0.01) \\
    cons. &       & 0.68*** & 0.83*** & 0.72*** &       & 0.18*** & 0.24*** & 0.21*** &       & 0.06*** & 0.11*** & 0.07*** \\
          &       & (0.01) & (0.01) & (0.01) &       & (0.01) & (0.01) & (0.00) &       & (0.00) & (0.01) & (0.00) \\
    \midrule
    obs. &       & 10,947 & 2,544 & 17,309 &       & 10,947 & 2,544 & 17,309 &       & 10,947 & 2,544 & 17,309 \\
    \bottomrule
    \bottomrule
    \end{tabular}%

\begin{minipage}{0.8\textwidth}
\footnotesize
\textit{Notes:} Standard errors in parentheses. *** $p<0.001$, ** $p<0.01$, * $p<0.05$.
\end{minipage}\label{tab:repli1}
\end{center}
\end{table}

Because G1 and G2 are nearly identical as incentives, we pool them into a \emph{low-compliance} encouragement $z_2$, while we treat G3 as a \emph{high-compliance} encouragement $z_1$. Although interventions G1 and G2 were independent of intervention G3, we cannot apply the 2SLS equivalence from Proposition~\ref{result:2sls} because of asymmetry. Therefore, we use $x^\Delta = z_1 - z_2$ as an instrument, $x^\Sigma = z_1 + z_2$ as a covariate, and the 2SLS representation in Proposition~\ref{result:joint2sls}. We note that $x^\times = z_1 z_2 = 0$, so it is omitted.

In this setting, a natural interpretation of defiant behavior is one of frustration or disengagement. Since the main goal of the intervention was to promote graduation, students who lost part of the cash transfers due to non-attendance may experience frustration, potentially leading to non-graduation, despite being otherwise likely to graduate and enroll in tertiary education in the absence of the intervention.

It is clear that $z_1$ generates more compliers, as the additional USD~300 payment tied explicitly to tertiary enrollment substantially increases the salience and perceived returns to post-secondary education. In contrast, defiant behavior is likely to be more prevalent under $z_2$, where the opportunity cost of non-compliance is lower. For such individuals, assignment to the low-compliance encouragement may reduce educational effort more severely, due to the absence of the additional bonus conditional on tertiary enrollment that is present in G3.

Since G3 directly affects tertiary enrollment, we cannot use tertiary enrollment as a final outcome variable. Instead, we estimate the effect \emph{tertiary enrollment $\Rightarrow$ tertiary graduation}. The exclusion restriction in this setting requires that the cash transfer interventions affect tertiary graduation only through their impact on tertiary enrollment. This is plausible because the transfers are conditional on contemporaneous school attendance and do not provide direct incentives or information related to tertiary completion beyond inducing enrollment.

Table~\ref{tab:results1} compares three IV specifications: using G3 as the only instrument, using a pooled \emph{any incentive} instrument, and using DIIV. Standard IV estimators based on either G3 alone or pooled instruments yield weak first stages and statistically insignificant second-stage estimates. In contrast, DIIV generates a strong first-stage F-statistic by exploiting the contrast between encouragements and recovers a sizable, positive, and statistically significant effect. The estimated weighted average effect is approximately 75\%, indicating that among students whose tertiary enrollment was induced by the intervention, roughly three-quarters went on to graduate. 

\begin{table}[!ht]
\caption{The effect of secondary school graduation on tertiary school enrollment}
\begin{center}
    \begin{tabular}{lcccc}
          &       & \multicolumn{3}{c}{Method} \\
\cmidrule{3-5}          &       & Single Instrument & Single Instrument & DIIV \\
          &       & $z_1$ & $z_{pool}$ & $x^{\Delta}$ \\
          &       & large incentive G3 & any incentive  & $z_1-z_2$ \\
\cmidrule{3-5}    Tertiary school enrollment (First Stage) &       &       &       &  \\
          &       &       &       &  \\
    instrument &       & 0.03* & -0.00 & 0.04*** \\
          &       & (0.02) & (0.01) & (0.01) \\
    \midrule
    first-stage F &       & 3.900 & 0.303 & 39.67 \\
          &       &       &       &  \\
    \multicolumn{5}{l}{Graduated from Tertiary School (Second Stage)} \\
          &       &       &       &  \\
    secondary graduation &       & 0.36  & -1.20 & 0.75*** \\
          &       & (0.34) & (2.86) & (0.12) \\
    \midrule
    observations &       & 2,544 & 17,309 & 17,309 \\
    \bottomrule
    \bottomrule
    \end{tabular}%
\begin{minipage}{0.6\textwidth}
\footnotesize
\textit{Notes:} Standard errors in parentheses. *** $p<0.001$, ** $p<0.01$, * $p<0.05$.
\end{minipage} \label{tab:results1}
\end{center}
\end{table}


\subsection{Example 2: Credit Market in South Africa}
\label{sec:app2_corr}

We next study advertising in the microfinance sector of South Africa, using data from \cite{bertrand2010s}. This study implements a large randomized controlled trial with $N=58{,}168$ potential borrowers receiving advertising letters that varied along multiple information and persuasion dimensions. The paper studies the effect of these interventions on $(i)$ application to the partner lender, $(ii)$ loan take-up at the partner lender, and $(iii)$ loan take-up from \emph{outside} lenders. Table~\ref{tab:bertrand_replication} reports a slight variation of Table~3 in \cite{bertrand2010s}.

\begin{table}[!ht]
\caption{Replication of \cite{bertrand2010s}}
\begin{center}
    \begin{tabular}{lrccc}
          &       & applied & took loan & took loan outside \\
\cmidrule{3-5}          &       &       &       &  \\
    low interest rate &       & 0.04*** & 0.04*** & 0.00 \\
          &       & (0.00) & (0.00) & (0.00) \\
    rate shown explicitly &       & -0.01*** & -0.01*** & -0.00 \\
          &       & (0.00) & (0.00) & (0.00) \\
    one payment example &       & -0.03*** & -0.02*** & -0.01* \\
          &       & (0.00) & (0.00) & (0.00) \\
    competitor's example shown &       & -0.00 & 0.00  & 0.01 \\
          &       & (0.00) & (0.00) & (0.00) \\
    no explicit usage example for loan &       & 0.01* & 0.00  & -0.00 \\
          &       & (0.00) & (0.00) & (0.00) \\
    photo gender match &       & -0.00 & -0.00 & 0.00 \\
          &       & (0.00) & (0.00) & (0.00) \\
    photo race match &       & 0.00  & -0.00 & -0.00 \\
          &       & (0.00) & (0.00) & (0.00) \\
    cellphone lottery prize mentioned &       & -0.00 & -0.00 & -0.00 \\
          &       & (0.00) & (0.00) & (0.00) \\
    Constant &       & 0.09*** & 0.07*** & 0.22*** \\
          &       & (0.00) & (0.00) & (0.01) \\
    \midrule
    Observations &       & 53,194 & 53,194 & 53,194 \\
    \bottomrule
    \bottomrule
    \end{tabular}%
\begin{minipage}{0.8\textwidth}
\footnotesize
\textit{Notes:} A dummy for interest rate $\leq$ 8\% was used instead of a continuous variable. Standard errors in parentheses. *** $p<0.001$, ** $p<0.01$, * $p<0.05$.
\end{minipage}\label{tab:bertrand_replication}
\end{center}
\end{table}

We regard loan take-up from other lenders as the final outcome $y$ and application to the partner lender as the treatment variable $d$. Under standard substitution logic, and because the vast majority of applicants got the load, an increase in applications to the partner lender should reduce borrowing from competitors, yielding a negative causal effect, which is the hypothesis we seek to test. We use two advertising nudges as instruments. The low-interest-rate indicator $z_1$ is the natural high-compliance encouragement. To construct a lower-compliance instrument, we use whether the letter included a competitor payment example. The letter either displayed or omitted an illustrative competitor repayment schedule; from the example, borrowers could infer an implied competitor interest rate. We define $z_2 :=$ ``competitor example not shown'' as the corresponding encouragement.

In  the field experiment, the competitor payment comparison shown to some borrowers was constructed using a monthly interest rate of 15\%. Evidence from a companion study using the same lender indicates that this rate lies just above the upper bound of the lender’s own pricing distribution \citep[][Appendix Table~2]{karlan2008credit}, rather than reflecting a central or typical benchmark. As a result, the competitor comparison likely overstated the cost of the relevant outside option faced by many borrowers and plausibly operated as a discouraging frame rather than as a neutral source of information.

This interpretation also has implications for behavioral responses when the competitor comparison is removed. If the status quo in lending advertising emphasizes unfavorable alternatives through a high-cost benchmark, then omitting this comparison constitutes an encouragement that moves away from that framing. While some borrowers may become more willing to apply once the discouraging comparison is removed, others may respond in the opposite direction. In particular, borrowers whose take-up was supported by the relative attractiveness created by the high-cost comparison may become less likely to apply once this framing is withdrawn. Thus, in addition to discouragement, the experimental design naturally allows for defiance with respect to the encouragement that omits the competitor example.

We are interested in estimating the causal effect \emph{application to partner lender $\Rightarrow$ obtain a loan from competitors}. The exclusion restriction requires that advertising interventions affect outside borrowing only through their impact on application to the partner lender. This restriction is likely satisfied because the mailings do not reduce search costs, alter eligibility, or expand credit supply in outside markets. Moreover, while one intervention arm includes a competitor payment example from which an interest rate can be inferred, this rate was plausibly perceived as unrepresentative, limiting its ability to directly affect borrowing decisions outside the partner lender.

Table~\ref{tab:results2} compares three IV specifications: using the low-interest-rate indicator alone, using a pooled \emph{any incentive} instrument ($\max\{z_1,z_2\}$), and using DIIV. Using either the strongest incentive alone or pooling instruments yields estimates that are either statistically insignificant or of the wrong sign relative to the hypothesized substitution effect (possibly due to large biases). In contrast, DIIV generates a strong first stage and a negative, economically meaningful estimate, consistent with substitution toward outside lenders when applications to the partner lender increase. In this setting, DIIV improves estimation by separating price-based encouragement from framing-based discouragement, thereby isolating the behavioral margin relevant for competitive borrowing.

\begin{table}[!ht]
\caption{Effect of application to the partner lender on taking a loan from competitors}
\begin{center}
    \begin{tabular}{lrccc}
          &       & \multicolumn{3}{c}{Method} \\
\cmidrule{3-5}          &       & Single Instrument & Single Instrument & DIIV \\
          &       & $z_1$ & $z_{pool}$ & $x^\Delta$ \\
          &       & low interest & $max\{z_1,z_2\}$ & $z_1 - z_2$ \\
\cmidrule{3-5}    \multicolumn{5}{l}{applied (First Stage)} \\
          &       &       &       &  \\
    instrument &       & 0.04*** & 0.04*** & 0.02*** \\
          &       & (0.00) & (0.00) & (0.00) \\
    \midrule
    First-stage F &       & 319.9 & 225.5 & 78.19 \\
          &       &       &       &  \\
    \multicolumn{5}{l}{took loan from outside (Second Stage)} \\
          &       &       &       &  \\
    applied &       & 0.10  & 0.33** & -0.39* \\
          &       & (0.08) & (0.10) & (0.17) \\
    \midrule
    observations &       & 58,168 & 58,168 & 58,168 \\
    \bottomrule
    \bottomrule
    \end{tabular}%
\begin{minipage}{0.6\textwidth}
\footnotesize
\textit{Notes:} Standard errors in parentheses. *** $p<0.001$, ** $p<0.01$, * $p<0.05$.
\end{minipage}
\label{tab:results2}
\end{center}
\end{table}

\subsection{Example 3: Persistence in Voter Turnout}
\label{sec:app3_voting}

We next study voter participation and persistence in turnout using two closely related field experiments on ballot secrecy and voter mobilization. The first study, \cite{gerber2008secrecy}, implements a large-scale randomized intervention in the context of the 2010 U.S. congressional election ($N=69{,}488$). Households received official letters sent on Secretary of State letterhead that varied in their informational content regarding election administration and ballot secrecy. The second study, \cite{gerber2012ballotsecrecy} ($N=3{,}744$), builds on the same experimental design to examine whether these interventions have persistent effects on participation in the subsequent 2012 elections.

The interventions in \cite{gerber2008secrecy} include letters explicitly addressing concerns about ballot secrecy, alongside two placebo communications that do not mention secrecy.\footnote{There are also letters addressing civic duty, which were not used in \cite{gerber2012ballotsecrecy} and are therefore excluded here.} The short placebo contains minimal text, while the long placebo matches the secrecy letters in length and emphasizes the role of the Secretary of State in election administration without providing reassurance about ballot secrecy. The primary outcome analyzed in that study is turnout in the 2010 congressional election.

Our object of interest follows the focus of \cite{gerber2012ballotsecrecy}. We study the causal effect \emph{voting in 2010 $\Rightarrow$ voting in 2012}, using the randomized ballot-related mailings as instruments. Turnout in 2010 constitutes the treatment variable, while the mailings serve as exogenous sources of variation in participation. Identification therefore relies on an exclusion restriction: the 2010 interventions must affect turnout in 2012 only through induced participation in 2010. Given the one-time nature of the mailing, the two-year gap between elections, and the absence of any reinforcement or forward-looking content, this assumption is plausible. In particular, it is unlikely that participation decisions in 2012 are directly influenced by recall of a ballot-secrecy letter received two years earlier. Instead, persistence in turnout is most naturally attributed to habit formation, consistent with the interpretation emphasized in \cite{gerber2012ballotsecrecy}.

Table~\ref{tab:turnout_replication} replicates the main reduced-form results from \cite{gerber2012ballotsecrecy}, taking the short placebo mailing as the baseline condition. The sample size is smaller due to data availability constraints. We observe positive reduced-form effects of the secrecy interventions on turnout in the 2012 primary election, while no statistically significant effects are detected for the general election.

\begin{table}[!ht]
\caption{Replication of \cite{gerber2012ballotsecrecy}. Reduced form of the 2010 intervention on 2012 voter turnout}
\begin{center}
    \begin{tabular}{lccccc}
          &       & presidential primary & congressional primary & presidential general & index (0 - 3) \\
\cmidrule{3-6}          &       &       &       &       &  \\
    any secrecy &       & 0.007* & 0.012** & 0.021 & 0.041* \\
          &       & (0.003) & (0.004) & (0.016) & (0.018) \\
    female &       & -0.002 & -0.001 & 0.033* & 0.030 \\
          &       & (0.003) & (0.005) & (0.016) & (0.018) \\
    democrat &       & 0.003 & 0.026*** & 0.046* & 0.074*** \\
          &       & (0.002) & (0.006) & (0.018) & (0.020) \\
    republican &       & 0.052*** & 0.072*** & 0.092*** & 0.217*** \\
          &       & (0.011) & (0.013) & (0.026) & (0.036) \\
    family size &       & 0.002 & 0.001 & 0.024** & 0.028** \\
          &       & (0.002) & (0.002) & (0.008) & (0.010) \\
    \midrule
    observations &       & 3,516 & 3,516 & 3,516 & 3,516 \\
    \bottomrule
    \bottomrule
    \end{tabular}%
\begin{minipage}{0.8\textwidth}
\footnotesize
\textit{Notes:} Covariates are gender (female = 1), registered Democrat, registered Republican, family size, and town indicators. Standard errors in parentheses. *** $p<0.001$, ** $p<0.01$, * $p<0.05$.
\end{minipage}
\label{tab:turnout_replication}
\end{center}
\end{table}

We define $z_1$ as an indicator for assignment to any ballot-secrecy letter and $z_2$ as an indicator for assignment to the long placebo mailing. Although there is no $(z_1,z_2)=(1,1)$ group, we exploit contrasts between these interventions using the 2SLS representation in Proposition~\ref{result:joint2sls}. This approach is appropriate because there is no pure no-contact control group, the intervention arms are not parallel in the sense of Section~\ref{sec:parallel}, and assignment probabilities differ across treatments.

While some individuals may be reassured by explicit statements about ballot secrecy and become more likely to vote, others may respond adversely to official communications that heighten the salience of electoral authorities or institutional oversight. These concerns are particularly relevant for the long placebo intervention, which increases the salience of the Secretary of State without addressing secrecy and may discourage participation among individuals with low institutional trust.

As a result, opposing participation shifts across interventions are plausible. Secrecy-focused letters are likely to generate compliers by alleviating privacy concerns, whereas the long placebo may induce defiance by emphasizing institutional authority without reassurance. Pooling these interventions therefore likely aggregates positive and negative first-stage responses, implicitly combining encouragement-induced participation with discouragement-driven abstention.

Table~\ref{tab:results3} compares three IV specifications: using the secrecy intervention alone, using a pooled \emph{any intervention} instrument, and using DIIV. Across all specifications, first-stage F-statistics are modest, reflecting the well-known difficulty of mobilizing voter participation and the substantial loss of statistical power induced by restricting the sample to individuals observed in both the 2010 and 2012 elections. In the DIIV specification, the first-stage F-statistic is approximately 6.8. While below conventional thresholds, this value implies that any finite-sample bias due to weak instruments is limited. In particular, an F-statistic of this magnitude corresponds to a relative bias on the order of 15--20\% of the OLS bias, which is modest in this setting given that OLS estimates of turnout persistence are themselves small and potentially attenuated by measurement error.

\begin{table}[!ht]
\caption{Turnout recurrence}
\begin{center}
    \begin{tabular}{lccccccc}
          &       & \multicolumn{6}{c}{Method} \\
\cmidrule{3-8}          &       & \multicolumn{2}{c}{Single Instrument} & \multicolumn{2}{c}{Single Instrument} & \multicolumn{2}{c}{DIIV} \\
          &       & \multicolumn{2}{c}{$z_1$} & \multicolumn{2}{c}{$z_{pool}$} & \multicolumn{2}{c}{$x^{\Delta}$} \\
          &       & \multicolumn{2}{c}{secrecy} & \multicolumn{2}{c}{secrecy or long placebo} & \multicolumn{2}{c}{$z_1-z_2$} \\
    \multicolumn{8}{l}{turnout in 2010 (First Stage)} \\
          &       &       &       &       &       &       &  \\
    instrument &       & 0.037* & 0.037* & 0.028 & 0.028 & 0.023** & 0.023** \\
          &       & (0.018) & (0.018) & (0.017) & (0.017) & (0.009) & (0.009) \\
    \midrule
    First-stage F &       & 4.391 & 4.391 & 2.617 & 2.617 & 6.762 & 6.762 \\
          &       &       &       &       &       &       &  \\
    \multicolumn{8}{l}{turnout in 2012 (Second Stage)} \\
          &       & presidential & congressional & presidential & congressional & presidential & congressional \\
          &       &       &       &       &       &       &  \\
    turnout in 2010 &       & 0.231 & 0.128 & 0.187 & 0.096 & 0.350* & 0.218* \\
          &       & (0.181) & (0.121) & (0.219) & (0.147) & (0.161) & (0.098) \\
    \midrule
    observations &       & 2,922 & 2,922 & 3,516 & 3,516 & 3,516 & 3,516 \\
    \bottomrule
    \bottomrule
    \end{tabular}%
\begin{minipage}{0.6\textwidth}
\footnotesize
\textit{Notes:} Covariates are gender (female = 1), registered Democrat, registered Republican, family size, and town indicators. Columns using DIIV also include $x^\Sigma = z_1 + z_2$ as a covariate, as stated in Proposition~\ref{result:joint2sls}. Standard errors in parentheses. *** $p<0.001$, ** $p<0.01$, * $p<0.05$.
\end{minipage}
\label{tab:results3}
\end{center}
\end{table}

Turning to the second stage, single-instrument IV estimates are positive but imprecise and sensitive to pooling choices, reflecting the aggregation of heterogeneous behavioral responses. In contrast, DIIV yields larger and statistically significant estimates for both presidential and congressional turnout in 2012. The point estimates imply that voting in 2010 increases the probability of voting again by approximately 35\% in presidential primaries and about 22\% in congressional primaries. In this setting, DIIV improves identification by separating reassurance-based encouragement from authority-salience-based discouragement, thereby isolating the behavioral margin most relevant for turnout persistence.


\section{Simulations}

Finally, we conduct simulations to clarify the interpretation of DIIV relative to standard overidentified IV when compliers and defiers coexist. We consider four simulation environments. Each environment consists of 1{,}000 independent trials, where each trial represents an experiment with a sample size of 10{,}000 observations.

The population is composed of four behavioral types: 10\% always-takers, 10\% never-takers, 40\% persuasion-prone units (potential compliers), and 40\% reactance-prone units (potential defiers). Treatment effects are heterogeneous across types, with potential outcome gains given by $\tau_A = 4$, $\tau_C = 3$, $\tau_F = 2$, and $\tau_N = 1$. Outcome noise is additive and normally distributed with variance one.

Two binary instruments, $z_1$ and $z_2$, are generated by thresholding correlated latent Gaussian variables. Specifically, $e_1 \sim \mathcal{N}(0,1)$, $e_2 = u + \rho e_1$ with $u \sim \mathcal{N}(0,1)$ and $\rho = -0.45$, and $z_j = \mathbbm{1}\{e_j>0\}$ for $j\in\{1,2\}$. Thus, each instrument equals one with probability $1/2$, and their dependence reflects the latent Gaussian correlation.

For always-takers, $d_i=1$, and for never-takers, $d_i=0$. For $\theta \in \{C,F\}$, treatment take-up follows
\[
d_i = \mathbbm{1}\{\kappa_{\theta 1} z_{i,1} + \kappa_{\theta 2} z_{i,2} > \eta_i\},
\]
where $\eta_i \sim \mathcal{N}(0,\sigma^2)$ independently of $(z_1,z_2)$.

We vary two dimensions across environments: (i) \textit{Responsiveness asymmetry:} either persuasion-prone units (type $C$) are more responsive than reactance-prone units (type $F$), or vice versa, implemented through larger absolute values of the corresponding coefficients $\kappa_{\theta j}$; and (ii) \textit{Behavioral overlap:} a high-noise case with $\sigma=2$ and a low-noise case with $\sigma=1$. Lower noise sharpens behavioral responses.

\begin{table}[!htbp]
\centering
\caption{Simulation parameter values across environments}
\label{tab:simparams}
\begin{tabular}{lccccc}
Environment & $\kappa_{C1}$ & $\kappa_{C2}$ & $\kappa_{F1}$ & $\kappa_{F2}$ & $\sigma$ \\
\midrule
Persuasion more responsive (high noise) 
    & 2.0 & 0.5 & -0.2 & -0.5 & 2 \\
Reactance more responsive (high noise) 
    & 0.5 & 0.2 & -0.5 & -2.0 & 2 \\
Persuasion more responsive (low noise) 
    & 2.0 & 0.5 & -0.2 & -0.5 & 1 \\
Reactance more responsive (low noise) 
    & 0.5 & 0.2 & -0.5 & -2.0 & 1 \\
\bottomrule
\bottomrule
\end{tabular}
\end{table}

For each environment, we estimate (i) standard overidentified 2SLS using $(z_1,z_2)$ as instruments and (ii) DIIV implemented via a just-identified 2SLS representation using the difference instrument $x^\Delta = z_1 - z_2$, controlling for $x^\Sigma = z_1 + z_2$ and $x^\times = z_1 z_2$.

Figure \ref{fig:simulation} reports the sampling distributions of both estimators. In all cases, DIIV lies within the expected range bounded by $\tau_F$ and $\tau_C$. In contrast, the overidentified IV estimator may fall outside this range, moving closer to $\tau_F$ when $p_F^2-p_F^1 > p_C^1-p_C^2$, and closer to $\tau_C$ when $p_C^1-p_C^2 > p_F^2-p_F^1$.

\begin{figure}[!htbp]
\begin{subfigure}[c]{0.45\textwidth}
    \centering
    \includegraphics[width=\linewidth]{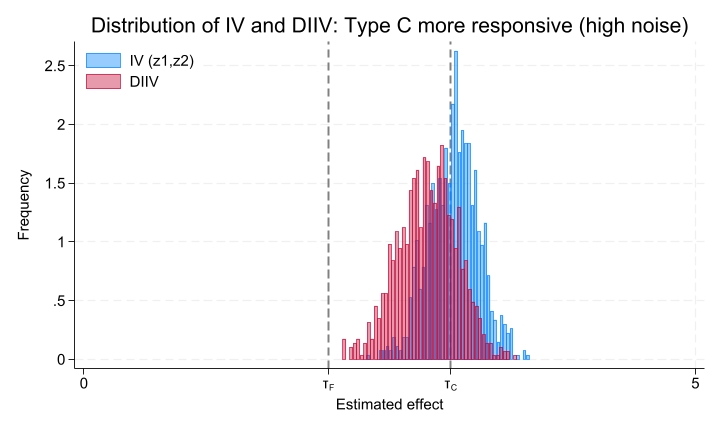}
    \caption{{Persuasion-prone more responsive (high noise).
    $p_C^1=0.137$, $p_C^2=0.039$; $p_F^1=0.016$, $p_F^2=0.039$.
    $\lambda = 0.805$.}}
\end{subfigure}
\hfill
\begin{subfigure}[c]{0.45\textwidth}
    \centering
    \includegraphics[width=\linewidth]{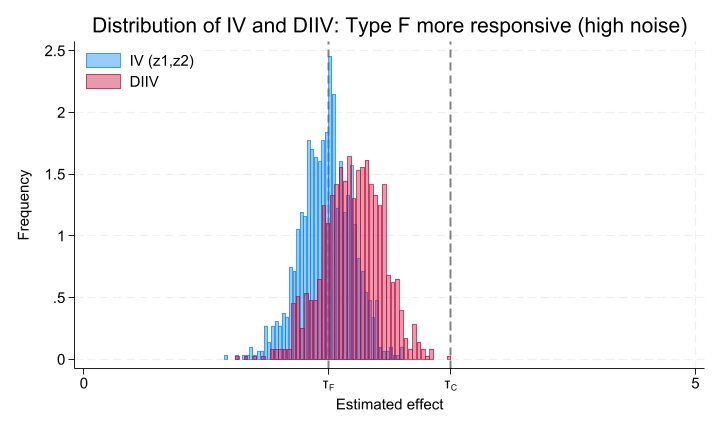}
    \caption{{Reactance-prone more responsive (high noise).
    $p_C^1=0.039$, $p_C^2=0.016$; $p_F^1=0.039$, $p_F^2=0.137$.
    $\lambda = 0.195$.}}
\end{subfigure}


\begin{subfigure}[c]{0.45\textwidth}
    \centering
    \includegraphics[width=\linewidth]{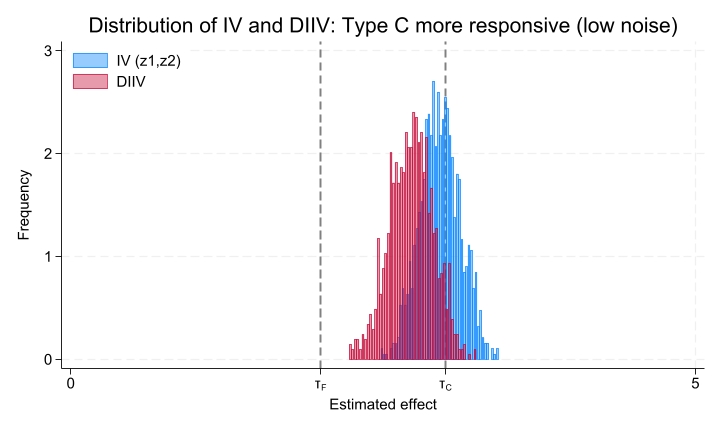}
    \caption{{Persuasion-prone more responsive (low noise).
    $p_C^1=0.191$, $p_C^2=0.077$; $p_F^1=0.032$, $p_F^2=0.077$.
    $\lambda = 0.718$.}}
\end{subfigure}
\hfill
\begin{subfigure}[c]{0.45\textwidth}
    \centering
    \includegraphics[width=\linewidth]{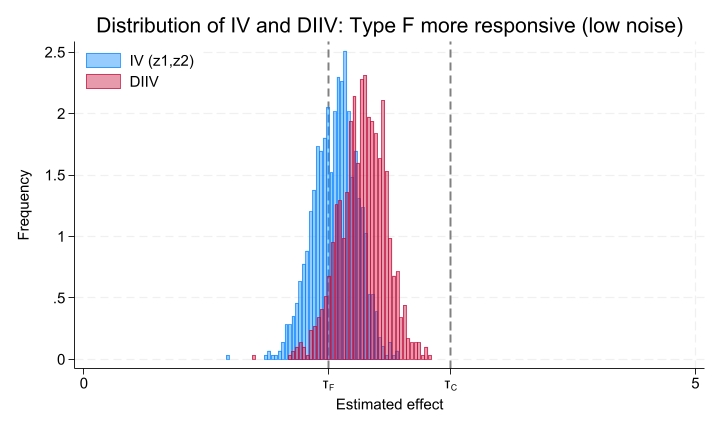}
    \caption{{Reactance-prone more responsive (low noise).
    $p_C^1=0.077$, $p_C^2=0.032$; $p_F^1=0.077$, $p_F^2=0.191$.
    $\lambda = 0.282$.}}
\end{subfigure}

\caption{{Sampling distributions of standard IV using $(z_1,z_2)$ and DIIV across four environments. Panels vary by which behavioral type is more responsive to the instruments, implemented through larger absolute coefficients $\kappa_{\theta}(m_j)$, and by the variance of the latent noise term $\eta_i$. Vertical dashed lines indicate $\tau_F = 2$ and $\tau_C = 3$. Reported probabilities $p_\theta^j$ correspond to type-$\theta$-specific shifts in treatment take-up induced by instrument $j$, and $\lambda$ denotes the DIIV weight on $\tau_C$ in the convex combination $\lambda \tau_C + (1-\lambda)\tau_F$.}}
\label{fig:simulation}
\end{figure}

The simulations highlight an interpretational distinction between DIIV and standard IV. DIIV is explicitly constructed as a convex combination of the marginal treatment effects for compliers and defiers, ensuring that the resulting estimand remains between $\tau_C$ and $\tau_F$ across all configurations. In contrast, IV with multiple instruments forms an implicit linear combination of the instruments’ first stages. When defiers are present, these implicit weights need not be positive, so the overidentified IV estimand may place disproportionate weight on one type and need not lie between $\tau_C$ and $\tau_F$.

Overall, the simulations show that DIIV delivers a stable and interpretable estimand precisely in environments where monotonicity violations are empirically relevant. By making the aggregation of complier and defier effects explicit and tied directly to observable shifts in treatment take-up, DIIV avoids the ambiguity that can arise from implicit weighting in overidentified IV designs.

\section{Conclusions}

DIIV is designed to incorporate the existence of defiers into the instrumental variables framework. By contrasting reduced forms and first stages across instruments with opposing behavioral orientation, DIIV identifies a convex combination of the marginal treatment effects on compliers and defiers. The weights in this combination reflect the relative net shifts in treatment take-up induced by each instrument, yielding a behaviorally interpretable estimand even when monotonicity fails.

Because DIIV admits a 2SLS-equivalent representation through simple data transformations, it can be implemented using standard estimation and inference tools. More broadly, the approach provides a transparent extension of IV to environments where distinct encouragements induce heterogeneous responses across individuals. In such settings, DIIV allows researchers to retain point identification and interpretability without ruling out empirically relevant forms of defiance.

\bibliographystyle{econ}
\bibliography{references}

@article{angrist1996identification,
  title={Identification of causal effects using instrumental variables},
  author={Angrist, Joshua D. and Imbens, Guido W. and Rubin, Donald B.},
  journal={Journal of the American Statistical Association},
  volume={91},
  number={434},
  pages={444--455},
  year={1996}
}

@article{BalkePearl1997,
  title={Bounds on treatment effects from studies with imperfect compliance},
  author={Balke, Alexander and Pearl, Judea},
  journal={Journal of the American Statistical Association},
  volume={92},
  number={439},
  pages={1171--1176},
  year={1997}
}

@article{Barrera2019medium,
  title={Medium- and long-term educational consequences of alternative conditional cash transfer designs: Experimental evidence from Colombia},
  author={Barrera-Osorio, Felipe and Linden, Leigh L. and Saavedra, Juan E.},
  journal={American Economic Journal: Applied Economics},
  volume={11},
  number={3},
  pages={54--91},
  year={2019}
}

@article{bertrand2010s,
  title={What's advertising content worth? Evidence from a consumer credit marketing field experiment},
  author={Bertrand, Marianne and Karlan, Dean and Mullainathan, Sendhil and Shafir, Eldar and Zinman, Jonathan},
  journal={The Quarterly Journal of Economics},
  volume={125},
  number={1},
  pages={263--306},
  year={2010}
}

@book{Brehm1966,
  title={A Theory of Psychological Reactance},
  author={Brehm, Jack W.},
  publisher={Academic Press},
  address={New York},
  year={1966}
}

@misc{christy2024countingdefiersdesignbasedmodel,
  title={Counting Defiers: A Design-Based Model of an Experiment Can Reveal Evidence Beyond the Average Effect},
  author={Christy, Neil and Kowalski, Amanda Ellen},
  year={2024},
  eprint={2412.16352},
  archivePrefix={arXiv},
  primaryClass={econ.EM},
  url={https://arxiv.org/abs/2412.16352}
}

@article{Dahl2023,
  author={Dahl, Christian M. and Huber, Martin and Mellace, Giovanni},
  title={It is never too LATE: A new look at local average treatment effects with or without defiers},
  journal={The Econometrics Journal},
  volume={26},
  number={3},
  pages={378--404},
  year={2023}
}

@article{DellaVigna2009,
  author={DellaVigna, Stefano},
  title={Psychology and Economics: Evidence from the Field},
  journal={Journal of Economic Literature},
  volume={47},
  number={2},
  pages={315--372},
  year={2009}
}

@article{dechaisemartin2017,
  author={de Chaisemartin, Cl\'ement},
  title={Tolerating defiance? Local average treatment effects without monotonicity},
  journal={Quantitative Economics},
  volume={8},
  number={2},
  pages={367--396},
  year={2017}
}

@article{gerber2008secrecy,
  title={Do perceptions of ballot secrecy influence turnout? Results from a field experiment},
  author={Gerber, Alan S and Huber, Gregory A and Doherty, David and Dowling, Conor M and Hill, Seth J},
  journal={American Journal of Political Science},
  volume={57},
  number={3},
  pages={537--551},
  year={2013},
  publisher={Wiley Online Library}
}

@article{gerber2012ballotsecrecy,
  title={Ballot secrecy concerns and voter mobilization: New experimental evidence about message source, context, and the duration of mobilization effects},
  author={Gerber, Alan S. and Huber, Gregory A. and Biggers, Daniel R. and Hendry, David J.},
  journal={American Politics Research},
  volume={42},
  number={5},
  pages={896--923},
  year={2014}
}

@article{hansen1982,
  author={Hansen, Lars Peter},
  title={Large sample properties of generalized method of moments estimators},
  journal={Econometrica},
  volume={50},
  number={4},
  pages={1029--1054},
  year={1982}
}

@article{huber2015,
  author={Huber, Martin and Mellace, Giovanni},
  title={Testing instrument validity for LATE identification based on inequality moment constraints},
  journal={The Review of Economics and Statistics},
  volume={97},
  number={2},
  pages={398--411},
  year={2015}
}

@article{imbens1994,
  author={Imbens, Guido W. and Angrist, Joshua D.},
  title={Identification and estimation of local average treatment effects},
  journal={Econometrica},
  volume={62},
  number={2},
  pages={467--475},
  year={1994}
}

@article{KamenicaGentzkow2011,
  author={Kamenica, Emir and Gentzkow, Matthew},
  title={Bayesian Persuasion},
  journal={American Economic Review},
  volume={101},
  number={6},
  pages={2590--2615},
  year={2011}
}

@article{karlan2008credit,
  title={Credit Elasticities in Less-Developed Economies: Implications for Microfinance},
  author={Karlan, Dean S. and Zinman, Jonathan},
  journal={American Economic Review},
  volume={98},
  number={3},
  pages={1040--1068},
  year={2008}
}

@article{kitagawa2015,
  author={Kitagawa, Toru},
  title={A test for instrument validity},
  journal={Econometrica},
  volume={83},
  number={5},
  pages={2043--2063},
  year={2015}
}

@article{kowalski2023behaviour,
  title={Behaviour within a clinical trial and implications for mammography guidelines},
  author={Kowalski, Amanda E.},
  journal={The Review of Economic Studies},
  volume={90},
  number={1},
  pages={432--462},
  year={2023}
}

@article{manski1990,
  author={Manski, Charles F.},
  title={Nonparametric Bounds on Treatment Effects},
  journal={American Economic Review},
  volume={80},
  number={2},
  pages={319--323},
  year={1990},
  note={Papers and Proceedings}
}

@article{MironBrehm2006,
  title={Reactance theory---40 years later},
  author={Miron, Anca M. and Brehm, Jack W.},
  journal={Zeitschrift f\"ur Sozialpsychologie},
  volume={37},
  number={1},
  pages={9--18},
  year={2006}
}

@article{mogstad2021causal,
  title={The causal interpretation of two-stage least squares with multiple instrumental variables},
  author={Mogstad, Magne and Torgovitsky, Alexander and Walters, Christopher R.},
  journal={American Economic Review},
  volume={111},
  number={11},
  pages={3663--3698},
  year={2021}
}

@article{mourifie2017,
  author={Mourifi\'e, Ismael and Wan, Yao},
  title={Testing local average treatment effect assumptions},
  journal={The Review of Economics and Statistics},
  volume={99},
  number={2},
  pages={305--313},
  year={2017}
}

@misc{noack2021,
  author={Noack, Claudia},
  title={Sensitivity of LATE Estimates to Violations of the Monotonicity Assumption},
  year={2021},
  eprint={2106.06421},
  archivePrefix={arXiv},
  primaryClass={econ.EM},
  url={https://arxiv.org/abs/2106.06421}
}

@techreport{richardson2010,
  author={Richardson, Thomas and Evans, Robin J. and Robins, James M.},
  title={Analysis of the Binary Instrumental Variable Model},
  institution={University of Washington, Center for Statistics and the Social Sciences},
  type={Working Paper},
  number={128},
  year={2010}
}

@article{small2017,
  author={Small, Dylan S. and Tan, Zhiqiang and Lorch, Scott},
  title={Instrumental Variable Estimation with a Stochastic Monotonicity Assumption},
  journal={Statistical Science},
  volume={32},
  number={4},
  pages={561--576},
  year={2017}
}

@article{swanson2015,
  author={Swanson, Sonja A. and Hern\'an, Miguel A. and Miller, Matthew and Robins, James M. and Richardson, Thomas S.},
  title={Definition and evaluation of the monotonicity condition for instrumental variable analyses},
  journal={Epidemiology},
  volume={26},
  number={3},
  pages={414--420},
  year={2015}
}

@article{TverskyKahneman1981,
  title={The framing of decisions and the psychology of choice},
  author={Tversky, Amos and Kahneman, Daniel},
  journal={Science},
  volume={211},
  number={4481},
  pages={453--458},
  year={1981}
}


\appendix
\section{Appendix}
\subsection{Summary of all results with proofs}
\setcounter{equation}{0} %
\setcounter{lemma}{0} %
\setcounter{proposition}{0} %
\setcounter{corollary}{0} 
\renewcommand{\theequation}{A\arabic{equation}}

\begin{lemma}\label{app:rffs}
Under Assumption \ref{assumption:exclu},
\begin{equation}
RF_j = p_C^j \tau_C - p_F^j \tau_F \text{ and }
FS_j = p_C^j - p_F^j.
\end{equation}
\end{lemma}

\begin{proof}[Proof of Lemma \ref{result:rffs}]
Fix a frame $j\in\{1,2\}$. Recall that $\theta$ defines types and $\eta_i$ further defines behavior, partitioning
\[
I_j = \{A_j,N_j,C_j^C,C_j^A,C_j^N,F_j^F,F_j^A,F_j^N\}.
\]

Given $j$, the instrument $z_j$ changes behavior only on $C_j^C$ and $F_j^F$. Define the events
$E_C(j):=\{\eta_i \in [\underline{v}-v_C(j),\underline{v})\}$ and
$E_F(j):=\{\eta_i\in[\underline{v},\underline{v}-v_F(j))\}$. Then:

\begin{eqnarray*}
RF_j &=& \mathbb{E}[y_i \mid z_{i,j}=1]-\mathbb{E}[y_i \mid z_{i,j}=0]\\
&=& \pi_C \mathbb{E}[y_i \mid z_{i,j}=1,\theta=C]
   + \pi_F \mathbb{E}[y_i \mid z_{i,j}=1,\theta=F] \\
&-& \pi_C \mathbb{E}[y_i \mid z_{i,j}=0,\theta=C]
   - \pi_F \mathbb{E}[y_i \mid z_{i,j}=0,\theta=F]\\
&=& \pi_C \phi_C^C(j)\mathbb{E}[y_i \mid z_{i,j}=1,\theta=C,E_C(j)]
   + \pi_F \phi_F^F(j)\mathbb{E}[y_i \mid z_{i,j}=1,\theta=F,E_F(j)]\\
&-& \pi_C \phi_C^C(j)\mathbb{E}[y_i \mid z_{i,j}=0,\theta=C,E_C(j)]
   - \pi_F \phi_F^F(j)\mathbb{E}[y_i \mid z_{i,j}=0,\theta=F,E_F(j)]\\
&=& \pi_C \phi_C^C(j)\left(\mathbb{E}[y_i \mid \theta=C,d_i=1]
   - \mathbb{E}[y_i \mid \theta=C,d_i=0]\right)\\
&-& \pi_F \phi_F^F(j)\left(\mathbb{E}[y_i \mid \theta=F,d_i=1]
   - \mathbb{E}[y_i \mid \theta=F,d_i=0]\right)\\
&=& \pi_C \phi_C^C(j)\mathbb{E}[Y(1)-Y(0)\mid \theta=C]
   - \pi_F \phi_F^F(j)\mathbb{E}[Y(1)-Y(0)\mid \theta=F]\\
&=& \pi_C \phi_C^C(j)\tau_C - \pi_F \phi_F^F(j)\tau_F\\
&=& p_C^j \tau_C - p_F^j \tau_F.
\end{eqnarray*}

The first equality is the definition of the reduced form. The second equality conditions on types and uses the fact that types $A$ and $N$ do not change behavior; therefore, by part (i) of Assumption \ref{assumption:exclu}, there is no differential in $y_i$. The third equality further conditions on realizations of $\eta_i$ and uses that take-up differs only on $E_C(j)$ and $E_F(j)$. The remaining equalities follow from the take-up rule, the definition of $y_i$, and Assumption \ref{assumption:exclu}. The same steps apply to $FS_j$.
\end{proof}


\begin{proposition}\label{app:maininden}
If Assumptions \ref{assumption:exo}, \ref{assumption:exclu} and \ref{assumption:main} hold, then the DIIV estimand identifies:
\begin{equation}\label{app:eqdiiv}
\tau_{DIIV} = \lambda \tau_C + (1-\lambda)\tau_F,
\end{equation}
where
\begin{equation}
\lambda
=
\frac{p_C^1 - p_C^2}{(p_C^1 - p_C^2)-(p_F^1 - p_F^2)}
\in [0,1].
\end{equation}
\end{proposition}

\begin{proof}
Since Assumption \ref{assumption:exclu} holds, by Lemma \ref{result:rffs},
\[
RF_1 - RF_2 = (p_C^1 - p_C^2)\tau_C - (p_F^1 - p_F^2)\tau_F,
\]
and
\[
FS_1 - FS_2 = (p_C^1 - p_C^2)-(p_F^1 - p_F^2).
\]
Since Assumption \ref{assumption:main} holds, $(p_C^1 - p_C^2)-(p_F^1 - p_F^2) >0$, therefore we can define:
\begin{equation}\label{app:lambda}
\lambda
=
\frac{p_C^1 - p_C^2}{(p_C^1 - p_C^2)-(p_F^1 - p_F^2)}
\in [0,1].
\end{equation}
Equation \ref{app:eqdiiv} follows. Moreover, it identifies a causal effect by Assumption \ref{assumption:exo}.
\end{proof}


\begin{corollary}\label{app:specialcases}
If Assumptions \ref{assumption:exo}, \ref{assumption:exclu} and \ref{assumption:main} hold, then:
(i) If $p_F^1 = p_F^2$, $\tau_{DIIV} = \tau_C$.
(ii) If $p_C^1 = p_C^2$, $\tau_{DIIV} = \tau_F$.
\end{corollary}

\begin{proof}
Both results follow directly from Equation \ref{app:lambda}.
\end{proof}


\begin{proposition}\label{app:fliploss}
Consider $z_2^+ = 1 - z_2$. Then
\begin{equation}
RF_2^+ = -RF_2 \text{ and } FS_2^+ = -FS_2.
\end{equation}
Define the instrument $z_{\mathrm{pool\&flip}} = z_1 + z_2^+$. The standard IV estimand using $z_{\mathrm{pool\&flip}}$ satisfies
\begin{equation}
\tau_{\mathrm{pool\&flip}}
=
\frac{RF_1 - RF_2}{FS_1 - FS_2}
=
\tau_{DIIV}.
\end{equation}
\end{proposition}

\begin{proof}
\begin{eqnarray*}
RF_2^+ &=& \mathbb{E}[y_i \mid z_2^+=1] - \mathbb{E}[y_i \mid z_2^+=0] \\
&=& \mathbb{E}[y_i \mid z_2=0] - \mathbb{E}[y_i \mid z_2=1] \\
&=& -RF_2.
\end{eqnarray*}
The same argument applies to $FS_2^+$.
\end{proof}


\begin{proposition}\label{app:2sls}
Let $h_i\in\{0,1\}$ indicate the framing, with $h_i=1$ for $j=1$ and $h_i=0$ for $j=2$. Let $z_{\text{pool}}=z_1+z_2\in\{0,1\}$ denote the within-frame assignment. Construct the composite instrument
\[
w_i = z_{\text{pool},i} h_i + (1-z_{\text{pool},i})(1-h_i).
\]
Estimate the 2SLS system
\begin{align}
\text{First stage:}\quad
& d_i = \alpha_d + \beta w_i + \gamma_d h_i + e_{d,i}, \\
\text{Second stage:}\quad
& y_i = \alpha_y + \tau \hat d_i + \gamma_y h_i + e_{y,i}.
\end{align}

Then, the 2SLS estimand equals
\begin{equation}
\tau_{\mathrm{2SLS}}
=
\frac{RF_1 - RF_2}{FS_1 - FS_2}
=
\tau_{\mathrm{DIIV}}.
\end{equation}
\end{proposition}

\begin{proof}
The proof mirrors that of Proposition \ref{result:joint2sls} and is therefore omitted.
\end{proof}


\begin{proposition}\label{app:general-diiv}
If Assumptions \ref{assumption:exo_prime}, \ref{assumption:exclu_prime} and \ref{assumption:main_prime} hold, then the DIIV estimand identifies:
\begin{equation}
\tau_{DIIV} = \lambda \tau_C^{\mathrm{abs}} + (1-\lambda)\tau_F^{\mathrm{abs}},
\end{equation}

where 

\begin{equation}
\lambda = \frac{p_C^1 - p_C^2}{(p_C^1 - p_C^2)-(p_F^1 - p_F^2)}\in [0,1]
\end{equation}

\end{proposition}

\begin{proof}
We follow in a similar case as in the proofs of Lemma \ref{result:rffs} and Proposition \ref{result:maininden}. Indeed, it given $z_{-j}=0$:

\begin{eqnarray*}
RF_{j}^{(0)}
&=& \mathbb{E}[y_{i} \mid z_{i,j}=1,z_{i,-j}=0]-\mathbb{E}[y_{i} \mid z_{i,j}=0,z_{i,-j}=0] \\
&=& \pi_C \mathbb{E}[y_{i} \mid z_{i,j}=1,z_{i,-j}=0,\theta=C] + \pi_F \mathbb{E}[y_{i} \mid z_{i,j}=1,z_{i,-j}=0,\theta=F] \\
&&-\pi_C \mathbb{E}[y_{i} \mid z_{i,j}=0,z_{i,-j}=0,\theta=C] - \pi_F \mathbb{E}[y_{i} \mid z_{i,j}=0,z_{i,-j}=0,\theta=F]\\
&=& \pi_C \varphi_C(j)\,\mathbb{E}[y_{i} \mid z_{i,j}=1,z_{i,-j}=0,\theta=C,E_C(j)] \\
&&-\pi_C \varphi_C(j)\,\mathbb{E}[y_{i} \mid z_{i,j}=0,z_{i,-j}=0,\theta=C,E_C(j)] \\
   &&+ \pi_F \varphi_F(j)\,\mathbb{E}[y_{i} \mid z_{i,j}=1,z_{i,-j}=0,\theta=F,E_F(j)] \\
   &&- \pi_F \varphi_F(j)\,\mathbb{E}[y_{i} \mid z_{i,j}=0,z_{i,-j}=0,\theta=F,E_F(j)] \\
&=& \pi_C \varphi_C(j)\left(\mathbb{E}[y_{i}\mid \theta=C,d_i=1]-\mathbb{E}[y_{i}\mid \theta=C,d_i=0]\right) \\
&&-\pi_F \varphi_F(j)\left(\mathbb{E}[y_{i}\mid \theta=F,d_i=1]-\mathbb{E}[y_{i}\mid \theta=F,d_i=0]\right) \\
&=& \pi_C \varphi_C(j)\left(\mathbb{E}[Y(1)-Y(0)\mid \theta=C]\right)
   - \pi_F \varphi_F(j)\left(\mathbb{E}[Y(1)-Y(0)\mid \theta=F]\right) \\
&=& \pi_C \varphi_C(j)\tau_C - \pi_F \varphi_F(j)\tau_F \\
&=& p_C^j \tau_C - p_F^j \tau_F,
\end{eqnarray*}

where equalities follow the same reasoning as in the proofs of Lemma \ref{result:rffs}, which relies of Assumption \ref{assumption:exclu_prime} and the third inequality is the one that uses the fact that $z_{-j} = 0$ to be able to use the definition of the conditional probabilities $\varphi_C(j)$ and $\varphi_F(j)$. The same applies to $FS_{j}^{(0)} = p_C^j - p_F^j$

Next, it follows that

\[
RF_1^{(0)} - RF_2^{(0)} = (p_C^1 - p_C^2  ) \tau_C - (p_F^1 - p_F^2  ) \tau_F,
\]

and 

\[
FS_1^{(0)} - FS_2^{(0)} = (p_C^1 - p_C^2  ) - (p_F^1 - p_F^2 ).
\]

Since Assumption \ref{assumption:main_prime} holds, $(p_C^1 - p_C^2)-(p_F^1 - p_F^2) >0$, therefore we can define:

\begin{equation*}
\lambda = \frac{p_C^1 - p_C^2}{(p_C^1 - p_C^2)-(p_F^1 - p_F^2)},    
\end{equation*}

Equation \ref{app:eqdiiv} follows. Moreover, it identifies a causal effect by Assumption \ref{assumption:exo_prime}.

\end{proof}


\begin{lemma}\label{app:aligned-edge-identity}
Let $s_1,s_2\in\{-1,+1\}$ be known constants and define the aligned instruments
\[
\tilde z_j := \frac{1-s_j}{2}+s_j z_j,\qquad j\in\{1,2\}.
\]
For $a,b\in\{0,1\}$ define the aligned-cell means
\[
y_{ab}:=\mathbb{E}[y_i\mid \tilde z_{1i}=a,\tilde z_{2i}=b],
\qquad
d_{ab}:=\mathbb{E}[d_i\mid \tilde z_{1i}=a,\tilde z_{2i}=b].
\]

Then,
\begin{equation}
\frac{(y_{10}-y_{00})-(y_{01}-y_{00})}{(d_{10}-d_{00})-(d_{01}-d_{00})}
=
\frac{s_1\,\RF_{1}^{(0)}-s_2\,\RF_{2}^{(0)}}{s_1\,\FS_{1}^{(0)}-s_2\,\FS_{2}^{(0)}}.
\end{equation}
\end{lemma}

\begin{proof}
Fix $j\in\{1,2\}$. Since $\tilde z_j=\frac{1-s_j}{2}+s_j z_j$ and $s_j\in\{-1,+1\}$, we have
$\tilde z_j=z_j$ if $s_j=+1$ and $\tilde z_j=1-z_j$ if $s_j=-1$. Hence, for any $b\in\{0,1\}$,
\begin{equation}\label{eq:aligned-j-contrast}
\mathbb{E}[y_i\mid \tilde z_{ji}=1,\tilde z_{-j,i}=b]
-
\mathbb{E}[y_i\mid \tilde z_{ji}=0,\tilde z_{-j,i}=b]
=
s_j\Big(
\mathbb{E}[y_i\mid z_{ji}=1, z_{-j,i}=b^z]
-
\mathbb{E}[y_i\mid z_{ji}=0, z_{-j,i}=b^z]
\Big),
\end{equation}
where $b^z\in\{0,1\}$ is the unique value such that $\tilde z_{-j}=b$ (equivalently, $b^z=b$ if
$s_{-j}=+1$ and $b^z=1-b$ if $s_{-j}=-1$). In particular, when $b=0$ we obtain the identities

\begin{align*}
y_{10}-y_{00}
&=
\mathbb{E}[y_i\mid \tilde z_{1i}=1,\tilde z_{2i}=0]
-
\mathbb{E}[y_i\mid \tilde z_{1i}=0,\tilde z_{2i}=0] \\
&=
s_1\Big(
\mathbb{E}[y_i\mid z_{1i}=1,z_{2i}=0]
-
\mathbb{E}[y_i\mid z_{1i}=0,z_{2i}=0]
\Big) \\
&=
s_1\,\RF_{1}^{(0)},
\end{align*}

\begin{align*}
y_{01}-y_{00}
&=
\mathbb{E}[y_i\mid \tilde z_{1i}=0,\tilde z_{2i}=1]
-
\mathbb{E}[y_i\mid \tilde z_{1i}=0,\tilde z_{2i}=0] \\
&=
s_2\Big(
\mathbb{E}[y_i\mid z_{1i}=0,z_{2i}=1]
-
\mathbb{E}[y_i\mid z_{1i}=0,z_{2i}=0]
\Big) \\
&=
s_2\,\RF_{2}^{(0)}.
\end{align*}

Applying the same argument with $y_i$ replaced by $d_i$ yields
\[
d_{10}-d_{00}=s_1\,\FS_{1}^{(0)},
\qquad
d_{01}-d_{00}=s_2\,\FS_{2}^{(0)}.
\]
Substituting these four equalities into the left-hand side of \eqref{eq:aligned-edge-diiv} gives
\[
\frac{(y_{10}-y_{00})-(y_{01}-y_{00})}{(d_{10}-d_{00})-(d_{01}-d_{00})}
=
\frac{s_1\,\RF_{1}^{(0)}-s_2\,\RF_{2}^{(0)}}{s_1\,\FS_{1}^{(0)}-s_2\,\FS_{2}^{(0)}},
\]
which is \eqref{eq:aligned-edge-diiv}.
\end{proof}


\begin{proposition}\label{app:joint2sls}

Let $(z_1,z_2)\in\{0,1\}^2$ and directive orientations $s_1,s_2\in\{-1,+1\}$ be known constants, where
$s_j=+1$ denotes an encouragement directive and $s_j=-1$ a discouragement directive. Define the
\emph{oriented} instruments
\[
\tilde z_j := \frac{1-s_j}{2}+s_j z_j,\ j\in\{1,2\},
\]
and the variables
\begin{eqnarray*}
    x^\Delta &:=& \tilde z_1 - \tilde z_2,\\
    x^\Sigma &:=& \tilde z_1 + \tilde z_2,\\
    x^\times &:=& \tilde z_1 \tilde z_2.
\end{eqnarray*}
Then, the 2SLS of $y_i$ on $d_i$, using instrument $x^\Delta$ and covariates $(x^\Sigma,x^\times)$ following
\begin{align}
\text{First stage:}\quad & d_i = \alpha_d + \beta x^\Delta + \gamma_d^\Sigma x^\Sigma + \gamma_d^\times x^\times + e_{d,i},\\
\text{Second stage:}\quad & y_i = \alpha_y + \tau \hat d_i + \gamma_y^\Sigma x^\Sigma + \gamma_y^\times x^\times + e_{y,i}.
\end{align}
estimates
\[
\tau_{\text{2SLS}}
=
\tau_{\DIIV}.
\]
\end{proposition}

\begin{proof}
We show that the 2SLS coefficient on $d_i$ in the regression of $y_i$ on $d_i$,
using $x_i^{\Delta}$ as the excluded instrument and
$(1,x_i^{\Sigma},x_i^{\Pi})$ as included exogenous covariates, equals the DIIV
edge contrast
\[
\tau_{DIIV}
=
\frac{\big(y_{10}-y_{00}\big)-\big(y_{01}-y_{00}\big)}
     {\big(d_{10}-d_{00}\big)-\big(d_{01}-d_{00}\big)},
\]
where
\[
y_{z_1z_2}:=\mathbb{E}[y_i\mid \tilde z_{1i}=z_1,\tilde z_{2i}=z_2],
\qquad
d_{z_1z_2}:=\mathbb{E}[d_i\mid \tilde z_{1i}=z_1,\tilde z_{2i}=z_2],
\]
and $\tilde z_{ji}:=\frac{1-s_j}{2}+s_j z_{ji}$ for $j\in\{1,2\}$.
We proceed in five steps.

\paragraph{Step 1: Residualizing with respect to $(x_i^{\Sigma},x_i^{\Pi})$.}
Let $\tilde z_{ji}:=\frac{1-s_j}{2}+s_j z_{ji}$ for $j\in\{1,2\}$.
Define the derived regressors
\[
x_i^{\Delta}:=\tilde z_{1i}-\tilde z_{2i},
\qquad
x_i^{\Sigma}:=\tilde z_{1i}+\tilde z_{2i},
\qquad
x_i^{\Pi}:=\tilde z_{1i}\tilde z_{2i}.
\]
Consider the 2SLS regression of $y_i$ on $d_i$ using $x_i^{\Delta}$ as an
instrument and $(1,x_i^{\Sigma},x_i^{\Pi})$ as controls.

By the Frisch--Waugh--Lovell theorem, the 2SLS coefficient equals the IV
coefficient obtained after residualizing $y_i$, $d_i$, and $x_i^{\Delta}$ with
respect to $(1,x_i^{\Sigma},x_i^{\Pi})$. Let
\[
\tilde y_i:=y_i-\mathbb{E}[y_i\mid x_i^{\Sigma},x_i^{\Pi}],
\quad
\tilde d_i:=d_i-\mathbb{E}[d_i\mid x_i^{\Sigma},x_i^{\Pi}],
\quad
\tilde x_i^{\Delta}:=x_i^{\Delta}-\mathbb{E}[x_i^{\Delta}\mid x_i^{\Sigma},x_i^{\Pi}].
\]
With a single instrument,
\[
\tau_{\mathrm{2SLS}}
=
\frac{\Cov(\tilde x_i^{\Delta},\tilde y_i)}
     {\Cov(\tilde x_i^{\Delta},\tilde d_i)}.
\]

Because residuals have mean zero conditional on $(x_i^{\Sigma},x_i^{\Pi})$,
\begin{align}
\Cov(\tilde x_i^{\Delta},\tilde y_i)
&=\mathbb{E}\!\left[\Cov(x_i^{\Delta},y_i\mid x_i^{\Sigma},x_i^{\Pi})\right], \\
\Cov(\tilde x_i^{\Delta},\tilde d_i)
&=\mathbb{E}\!\left[\Cov(x_i^{\Delta},d_i\mid x_i^{\Sigma},x_i^{\Pi})\right].
\end{align}
Hence,
\begin{equation}
\tau_{\mathrm{2SLS}}
=
\frac{\mathbb{E}\!\left[\Cov(x_i^{\Delta},y_i\mid x_i^{\Sigma},x_i^{\Pi})\right]}
     {\mathbb{E}\!\left[\Cov(x_i^{\Delta},d_i\mid x_i^{\Sigma},x_i^{\Pi})\right]}.
\label{eq:tau2sls-conditional}
\end{equation}

\paragraph{Step 2: Design cells.}
The joint assignment $(\tilde z_{1i},\tilde z_{2i})$ takes values in $\{0,1\}^2$ with
probabilities
\[
q_{00}:=\Pr(\tilde z_1=0,\tilde z_2=0),
\quad
q_{10}:=\Pr(\tilde z_1=1,\tilde z_2=0),
\quad
q_{01}:=\Pr(\tilde z_1=0,\tilde z_2=1),
\quad
q_{11}:=\Pr(\tilde z_1=1,\tilde z_2=1).
\]

The induced values of $(x^{\Delta},x^{\Sigma},x^{\Pi})$ are
\[
\begin{array}{c|c|c|c}
(\tilde z_1,\tilde z_2) & x^{\Delta} & x^{\Sigma} & x^{\Pi} \\ \hline
(0,0) & 0 & 0 & 0 \\
(1,0) & 1 & 1 & 0 \\
(0,1) & -1 & 1 & 0 \\
(1,1) & 0 & 2 & 1
\end{array}
\]

\paragraph{Step 3: Conditional covariance with $d_i$.}
In cells $(x^{\Sigma},x^{\Pi})=(0,0)$ and $(2,1)$, $x^{\Delta}=0$
deterministically, so the covariance is zero.

In the edge cell $(x^{\Sigma},x^{\Pi})=(1,0)$,
\[
\Pr(\tilde z_1=1,\tilde z_2=0\mid x^{\Sigma}=1,x^{\Pi}=0)
=
\frac{q_{10}}{q_{10}+q_{01}},
\qquad
\Pr(\tilde z_1=0,\tilde z_2=1\mid x^{\Sigma}=1,x^{\Pi}=0)
=
\frac{q_{01}}{q_{10}+q_{01}}.
\]
Thus,
\[
\Cov(x^{\Delta},d\mid x^{\Sigma}=1,x^{\Pi}=0)
=
\frac{2q_{10}q_{01}}{(q_{10}+q_{01})^2}(d_{10}-d_{01}).
\]
Taking expectations,
\begin{equation}
\mathbb{E}\!\left[\Cov(x^{\Delta},d\mid x^{\Sigma},x^{\Pi})\right]
=
\frac{2q_{10}q_{01}}{q_{10}+q_{01}}(d_{10}-d_{01}).
\label{eq:EXD}
\end{equation}

\paragraph{Step 4: Conditional covariance with $y_i$.}
By the same argument,
\begin{equation}
\mathbb{E}\!\left[\Cov(x^{\Delta},y\mid x^{\Sigma},x^{\Pi})\right]
=
\frac{2q_{10}q_{01}}{q_{10}+q_{01}}(y_{10}-y_{01}).
\label{eq:EXY}
\end{equation}

\paragraph{Step 5: Equality with the DIIV estimand.}
Substituting \eqref{eq:EXY} and \eqref{eq:EXD} into
\eqref{eq:tau2sls-conditional} yields
\[
\tau_{\mathrm{2SLS}}
=
\frac{y_{10}-y_{01}}{d_{10}-d_{01}}.
\]

Moreover, the definition $\tilde z_j=\frac{1-s_j}{2}+s_j z_j$ implies
\[
y_{10}-y_{00}=s_1\big(\mathbb{E}[y_i\mid z_1=1,z_2=0]-\mathbb{E}[y_i\mid z_1=0,z_2=0]\big)
=s_1\,\RF_1^{(0)},
\]
\[
y_{01}-y_{00}=s_2\big(\mathbb{E}[y_i\mid z_1=0,z_2=1]-\mathbb{E}[y_i\mid z_1=0,z_2=0]\big)
=s_2\,\RF_2^{(0)},
\]
and analogously $d_{10}-d_{00}=s_1\,\FS_1^{(0)}$ and $d_{01}-d_{00}=s_2\,\FS_2^{(0)}$.
Therefore,
\[
\tau_{\DIIV}
=
\frac{s_1\,\RF_1^{(0)}-s_2\,\RF_2^{(0)}}
     {s_1\,\FS_1^{(0)}-s_2\,\FS_2^{(0)}}
=
\frac{(y_{10}-y_{00})-(y_{01}-y_{00})}{(d_{10}-d_{00})-(d_{01}-d_{00})}
=
\frac{y_{10}-y_{01}}{d_{10}-d_{01}},
\]
so $\tau_{\mathrm{2SLS}}=\tau_{DIIV}$.

This equality holds for arbitrary
$q_{00},q_{10},q_{01},q_{11}$ with $q_{10},q_{01}>0$.
\end{proof}

\end{document}